\documentclass{l4dc2025}

\jmlrvolume{}          
\jmlrpages{}           
\jmlryear{}            
\jmlrproceedings{}{}   

\title[Near Optimal Convergence to Coarse Correlated Equilibrium in Markov Games]{Near Optimal Convergence to Coarse Correlated Equilibrium in General-Sum Markov Games}
\usepackage{times}
\usepackage{amsmath, amssymb, amsfonts, mathtools}
\usepackage{graphicx,bm}
\usepackage{algorithm}
\usepackage{algorithmic}
\usepackage{bm}
\usepackage{tikz} 
\usepackage{pgfplots} 
\pgfplotsset{compat=1.17} 
\usepackage{ulem}
\usetikzlibrary{shapes.geometric, arrows} 
\usepackage{comment}
\usepackage{booktabs}
\usepackage{array} 
\usepackage[export]{adjustbox}
\expandafter\def\expandafter\normalsize\expandafter{%
    \normalsize%
    \setlength\abovedisplayskip{2pt}%
    \setlength\belowdisplayskip{2pt}%
    \setlength\abovedisplayshortskip{4pt}%
    \setlength\belowdisplayshortskip{0pt}%
}

\author{%
 \Name{Asrın Efe Yorulmaz} \Email{ay20@illinois.com}\\
 \addr University of Illinois, Urbana-Champaign
 \AND
 \Name{Tamer Başar} \Email{basar1@illinois.edu}\\
 \addr  University of Illinois, Urbana-Champaign
}

\begin{document}

\maketitle

\begin{abstract}
No-regret learning dynamics play a central role in game theory, enabling decentralized convergence to equilibrium for concepts such as Coarse Correlated Equilibrium (CCE) or Correlated Equilibrium (CE). In this work, we improve the convergence rate to CCE in general-sum Markov games, reducing it from the previously best-known rate of $\mathcal{O}(\log^5 T / T)$ to a sharper $\mathcal{O}(\log T / T)$. This matches the best known convergence rate for CE in terms of $T$, number of iterations, while also improving the dependence on the action set size from polynomial to polylogarithmic—yielding exponential gains in high-dimensional settings. Our approach builds on recent advances in adaptive step-size techniques for no-regret algorithms in normal-form games, and extends them to the Markovian setting via a stage-wise scheme that adjusts learning rates based on real-time feedback. We frame policy updates as an instance of Optimistic Follow-the-Regularized-Leader (OFTRL), customized for value-iteration-based learning. The resulting self-play algorithm achieves, to our knowledge, the fastest known convergence rate to CCE in Markov games.
\end{abstract}

\begin{keywords}%
  Learning in games, reinforcement learning, coarse correlated equilibrium, no-regret learning%
\end{keywords}

\section{Introduction}

Multi-agent systems are increasingly at the forefront of real-world applications, from autonomous driving \cite{shalevshwartz2016safemultiagentreinforcementlearning}, smart grids \cite{Chen_2022} and from LLMs \cite{wan2025remalearningmetathinkllms, park2025maporlmultiagentpostcotrainingcollaborative} to distributed robotics \cite{levine2016learninghandeyecoordinationrobotic} and financial markets \cite{zhang2023optimizingtradingstrategiesquantitative}. In these environments, agents must learn to make sequential decisions while accounting for the presence of other strategic agents whose actions affect the shared outcome. This interplay between individual learning and collective behavior gives rise to a central class of problems known as \emph{multi-agent reinforcement learning} (MARL) \cite{zhang2021multiagentreinforcementlearningselective}. The rise of MARL has been motivated not only by its broad applicability but also by the theoretical challenge of designing decentralized algorithms that ensure meaningful long-term behavior in interactive settings.

To model such scenarios, \emph{Markov games}—also known as stochastic games— offer a principled generalization of both Markov decision processes and normal-form games \cite{Shapleystochastic, LITTMAN1994157} . These games capture the temporal evolution of state, the strategic nature of agent interactions, and the dependence of rewards on joint actions. As such, they provide a natural framework for analyzing multi-agent learning dynamics in dynamic environments. However, understanding what equilibria and how fast these equilibria emerge under different decentralized learning methods in Markov games remains an open and fundamental question.

\begin{table}[ht]
\renewcommand{\arraystretch}{1.5} 
\centering
\caption{Comparison of convergence rates in normal-form and Markov games}
\label{tab:equilibria_comparison}
\begin{tabular}{|m{1.7cm}|>{\centering\arraybackslash}m{7cm}|>{\centering\arraybackslash}m{5.5cm}|}
\hline
\textbf{Aimed Equilibria} & \textbf{Normal-form Games} & \textbf{Markov Games} \\
\hline
NE & \( \!\mathcal{O}\Big(\frac{\log\mathcal|{A}_{\max}|\log T}{T}\Big)\) \cite{daskalakis2011} &  \(\!\mathcal{O}\Big(\frac{\log\mathcal|{A}_{\max}|}{T}\Big)\) \cite{yang2023ot1convergenceoptimisticfollowtheregularizedleadertwoplayer} \\
\hline
CE & \( \!\mathcal{O}\Big(\frac{\mathcal|{A}_{\max}|^{2.5}\log T}{T}\Big)\) \cite{anagnostides2022uncoupledlearningdynamicsolog} &  \( \!\mathcal{O}\Big(\frac{\mathcal|{A}_{\max}|^{2.5}\log T}{T}\Big)\) \cite{mao2024widetildeot1convergencecoarsecorrelated} \\
\hline 
CCE & \( \!\mathcal{O}\Big(\frac{(\log\mathcal|{A}_{\max}|)^2\log T}{T}\Big)\) \cite{Soleymani_2025} &  \( \!\mathcal{O}\Big(\frac{(\log\mathcal|{A}_{\max}|)^2\log T}{T}\Big)\) Theorem~\ref{thm:dlrc-regret-text} \\
\hline 
\end{tabular}
\end{table}

In normal-form games (NFGs), it is well established that when all players follow no-regret algorithms with $\mathcal{O}(\sqrt{T})$ regret against adversarial opponents, their joint play converges to an $\mathcal{O}(1/\sqrt{T})$-approximate equilibrium—specifically, a Nash equilibrium (NE) in the two-player zero-sum case and a (coarse) correlated equilibrium (CCE) CE, in general-sum settings \citep{hart2000simple, cesa2006prediction}. It is well known that multiplicative weights update (MWU) \cite{yoavschapire}, online mirror descent (OMD) \cite{nemirovski1983problem} and follow-the-regularized leader (FTRL) algorithms \cite{abernethy2008competing} all fall within this category. 

Although the regret for adversarial case is non-improvable, recent advances have further sharpened convergence rates for self-play algorithms in NFGs. Initiated by the seminal work of \citet{daskalakis2011}, which established a convergence rate of $\tilde{\mathcal{O}}(1/T)$ to NE in the two-player zero-sum setting, subsequent studies \citep{rakhlin2013optimizationlearninggamespredictable, syrgkanis2015fastconvergenceregularizedlearning, daskalakis2023nearoptimalnoregretlearninggeneral, anagnostides2022uncoupledlearningdynamicsolog, Soleymani_2025} provided more refined analyses and faster convergence guarantees compared to the baseline rate of $\tilde{\mathcal{O}}(1/\sqrt{T})$ for general NFGs. In particular, \citet{syrgkanis2015fastconvergenceregularizedlearning} demonstrated that when all players in a NFG employ the Optimistic FTRL (OFTRL) algorithm, their strategies converge to a CCE at a rate of $\mathcal{O}(T^{-3/4})$, enabled by the regret bounded by Variation in Utilities (RVU) framework. More recently, \citet{Soleymani_2025} established the best-known convergence rate to CCE in NFGs to date, namely, $\mathcal{O}\big(\log T / T\big)$, in the self-play setting. They showed that adaptive learning rule corresponds to a specific instantiation of the OFTRL algorithm with a tailored regularizer, which facilitates the improved convergence bound.

Recent advances in learning theory have brought attention to achieve faster convergence rates in Markov games, extending the foundational $\tilde{\mathcal{O}}(1/T)$ convergence results from NFGs to dynamic multi-agent environments. Notable progress has been made for various equilibrium concepts: NE in two-player zero-sum Markov games \cite{yang2023ot1convergenceoptimisticfollowtheregularizedleadertwoplayer}, CE in general-sum Markov games \cite{cai2024nearoptimalpolicyoptimizationcorrelated,mao2024widetildeot1convergencecoarsecorrelated}, and CCE in general-sum Markov games \cite{mao2024widetildeot1convergencecoarsecorrelated}. However, despite this progress, there remain significant limitations in the existing CCE convergence analyses. In particular, the best known rate for CCE learning, $\mathcal{O}((\log T)^5/T)$, lagged behind the CE convergence rate of $\mathcal{O}(\log T/T)$, a gap that undermines the appeal of CCE. Previous CCE results, such as those in~\cite{mao2024widetildeot1convergencecoarsecorrelated}, were based on stage-based frameworks that required long and inflexible stage lengths—typically of the order $T \gg Cn (\log T)^4$—as inherited from~\cite{daskalakis2023nearoptimalnoregretlearninggeneral}, where $C$ is a very large constant (see Lemmas~4.2 and C.4 in~\cite{daskalakis2023nearoptimalnoregretlearninggeneral}).

In this work, we resolve these issues by introducing a self-play algorithm for general-sum Markov games that achieves a CCE convergence rate of $\mathcal{O}(\log T / T)$, thereby closing the gap with the CE literature in terms of time-horizon dependence. Crucially, our method not only accelerates convergence in $T$, but also achieves exponential improvement in dependence on the action space size compared to best known CE convergence rate. This makes CCE learning computationally viable in high-dimensional settings where computing CE is often infeasible. As CCE encompasses a richer set of decentralized strategies and allows for more flexible agent behavior—particularly in environments where correlation is difficult to coordinate—our results enhance both the appeal of learning coarse equilibria in Markov games. Our approach builds on a \emph{dynamic step-size adaptation} scheme, which was proposed by \cite{Soleymani_2025} and shown to be equivalent to a particular instantiation of the OFTRL algorithm with a regularizer satisfying key smoothness properties. By establishing a RVU inequality under time-varying step sizes, and coupling it with value iteration procedure tailored to the episodic Markov games, we derive aforementioned convergence bounds.

\section{Preliminaries}

\subsection{Multi-player General-sum Markov Games}

We consider an \( N \)-player episodic Markov game defined by the tuple \( \mathcal{G} = \big([N], H, \mathcal{S}, \{ \mathcal{A}_i \}_{i=1}^N,\) \(\{ r_i \}_{i=1}^N, \{ P_h \}_{h=1}^{H} \big) \), where \( [N] := \{1,\dots, N\} \) denotes the set of players, \( H \in \mathbb{N}^+ \) is the episode length (horizon), \( \mathcal{S} \) is a finite state space, \( \mathcal{A}_i \) is the finite action set of player \( i \), and \( \mathcal{A} := \prod_{i=1}^N \mathcal{A}_i \) is the joint action space. The per-step reward function for player \( i \) is given by \( r_i : [H] \times \mathcal{S} \times \mathcal{A} \to [0,1] \), and the transition dynamics at step \( h \in [H] \) are specified by \( P_h : \mathcal{S} \times \mathcal{A} \to \Delta(\mathcal{S}) \). The agents interact in an unknown environment for $T$ episodes. At each step \( h \in [H] \), the system is in state \( s_h \in \mathcal{S} \). Each player selects an action \( a_{i,h} \in \mathcal{A}_i \), resulting in the joint action \( a_h = (a_{1,h}, \dots, a_{N,h}) \in \mathcal{A} \). Player \( i \) receives reward \( r_{i,h}(s_h, a_h) \), and the next state \( s_{h+1} \sim P_h(\cdot \mid s_h, a_h) \) is sampled. We make the assumption of, \cite{SongMeiBai2022,JinLiuWangYu2022, mao23ccedecentralized}, the episode beginning at a fixed initial state \( s_1 \in \mathcal{S} \). Finally, we let \( S = |\mathcal{S}|, A_i = |\mathcal{A}_i|, A_{\text{max}} = \max_{i\in [N]} A_i \). 

\subsection{Policies and Value Functions}

A (Markov) policy for player \( i \in [N] \) is a sequence of functions \( \pi_i = \{ \pi_{i,h} \}_{h=1}^H \), where each \( \pi_{i,h} : \mathcal{S} \to \Delta(\mathcal{A}_i) \) assigns a distribution over actions based on the current state at step \( h \). A joint policy \( \pi = (\pi_1, \dots, \pi_N) \) specifies a probability measure over the trajectory of states and joint actions. We write \( \pi = (\pi_i, \pi_{-i}) \) to distinguish player \( i \)'s policy from the other players. Given a policy \( \pi \), we define the V-function and  Q-function for player \( i \) at step \( h \in [H] \) and state \( s \in \mathcal{S} \) as
\begin{align}
V_{i,h}^\pi(s)\! := \!\mathbb{E}^\pi \!\Big[ \sum_{h' = h}^H \!r_{i,h'}(s_{h'}, a_{h'}) \!\,\Big|\, \!s_h \!= \!s \Big]\!,
Q_{i,h}^\pi(s, a)\! :=\! \mathbb{E}^\pi \!\Big[ \sum_{h' = h}^H \!r_{i,h'}(s_{h'}, a_{h'}) \!\,\Big|\, \!s_h \!= \!s,\, a_h\! = \!a \Big]
\end{align}
where \( a \in \mathcal{A} \) is a joint action taken at state \( s \). For any \( V_{i,h} \), we define the one-step Bellman operator as \( [P_h V](s, a) := \mathbb{E}_{s' \sim P_h(\cdot \mid s,a)}[V(s')]\), and for any \( Q_{i,h} \), we define the expected values under policies as \(
[Q_{i,h} \pi_h](s) := \langle Q_{i,h}(s, \cdot),\, \pi_h(\cdot|s) \rangle \) and \( [Q_{i,h} \pi_{-i,h}](s, a_i) := \langle Q_{i,h}(s, a_i, \cdot),\, \pi_{-i,h}(\cdot|s) \rangle. \)

\subsection{Decentralized Information Feedback}

In our setting, we assume that each player~\( i \) has access to the necessary information to compute their value function~\( V_{i,h}(s) \) at each stage~\( h \) and state~\( s \). The update for~\( V_{i,h}(s) \) depends only on the expected value of~\( r_h + P_h V_{i,h+1} \) under the joint policy of the other players, \( \pi_{-i,h}(\cdot \mid s) \). In particular, for any fixed~\( (s, a_i) \), we assume access to a \emph{reward oracle} that returns~\( \mathbb{E}_{a_{-i} \sim \pi_{-i,h}}[r_i(s, a_i, a_{-i})] \), and a \emph{transition oracle} that returns the marginal distribution~\( \mathbb{E}_{a_{-i} \sim \pi_{-i,h}}[P_h(\cdot \mid s, a_i,a_{-i})]\).

\subsection{Correlated Policies and Coarse Correlated Equilibrium}

We now extend our policy class to allow for coordination through shared randomness.  A \emph{correlated policy} \(\pi\) comprises a sequence of decision rules \( 
\pi_h : \Omega \times (\mathcal{S}\times\mathcal{A})^{h-1} \times \mathcal{S} \;\to\; \Delta(\mathcal{A}), h=[H], \) where \(\Omega\) is the space of random seeds.  At the start of an episode, a seed \(\omega\in\Omega\) is drawn.  Then, at each time step \(h\), given the current state \(s_h\) and the history \((s_1,a_1,\dots,s_{h-1},a_{h-1})\), the joint action is generated by \( a_h \;\sim\; \pi_h\bigl(\cdot \,\big|\, \omega,\,(s_1,a_1,\dots,s_{h-1},a_{h-1}),\,s_h\bigr).\)
Because the same \(\omega\) is used throughout the episode, this mechanism can induce arbitrary correlation across players’ choices, 
For a given correlated policy \(\pi\), denote by \(\pi_{-i}\) the marginal strategy of all players except \(i\).  We then define player \(i\)’s  \emph{best-response value}, against \(\pi_{-i}\) as
\( 
V_{i,1}^{\dagger,\pi_{-i}}(s_1)
\;:=\;\sup_{\substack{\pi'_i}}\; 
V_{i,1}^{(\pi_i',\,\pi_{-i})}(s_1),
\) 
where the supremum is over all (non-Markov) policies for player \(i\). Furthermore, since computing a Nash equilibrium is known to be PPAD-hard in general, the computation of Markov coarse correlated equilibria has become a central focus in the literature. We therefore provide its definition below:

\begin{definition}[\(\varepsilon\)-Coarse Correlated Equilibrium]
A correlated policy \(\pi\) is an \(\varepsilon\)-approximate Markov coarse correlated equilibrium if, for every player \(i\in [N]\), \( 
V_{i,1}^\pi(s_1)\;\ge\;V_{i,1}^{\dagger,\pi_{-i}}(s_1)\;-\;\varepsilon.\)
\end{definition}

\subsection{Regret and Learning Feedback}

We first recall the notion of \emph{external regret}. Let \( \mathcal{A}_i \) be the finite action set of player \( i \), and \( \Delta(\mathcal{A}_i) \) its probability simplex. At each round \( t \), the agent selects a mixed strategy \( x_i^{(t)} \in \Delta(\mathcal{A}_i) \), receives a utility vector \( \nu_i^{(t)} \in \mathbb{R}^{|\mathcal{A}_i|} \), and earns payoff \( \langle x_i^{(t)}, \nu_i^{(t)} \rangle \). The \emph{external regret} over \( T \) rounds is
\begin{align} \label{def:regret}
\mathrm{Reg}_i^T := \max_{a_i \in \mathcal{A}_i} \sum_{t=1}^T \big[ \langle a_i, \nu_i^{(t)} \rangle - \langle x_i^{(t)}, \nu_i^{(t)} \rangle \big],
\end{align}
which compares the agent's policy to the best fixed action in hindsight. An algorithm has no external regret if \( \mathrm{Reg}_i^T = o(T) \). We now extend this notion to Markov games, where each round \( t \), each player \( i \in [N] \) and step \( h \in [H] \) observes the expected utility vector,
\(  \nu_{i,h}^{(t)}(s, \cdot) := \big[Q_{i,h}^{(t)} \pi_{-i,h}^{(t)}\big](s, \cdot),\)  \(\forall s \in \mathcal{S} \). 
For each \((s,h)\) pair, the \emph{weighted external regret} incurred by player \( i \) is defined as
\begin{align}
\mathrm{reg}_{i,h}^t(s) := \max_{\pi^{\dagger,j}_{i,h} \in \mathcal{A}_i} \sum_{j=1}^t \alpha_t^j \ \left\langle \pi^{\dagger,j}_{i,h}-\pi^j_{i,h} ,  \left[Q_{i,h}^{(j)} \pi_{-i,h}^{(j)}\right](s, \cdot)  \right\rangle
\end{align}
where \( \{\alpha_j^t\}_{j=1}^t \) is a set of non-negative weights summing to one, and \( \pi_{i,h}^{\dagger,j} \) is player \( i \)'s best response to \( \pi_{-i,h}^{(j)} \) at step \( h \) in round \( j \). Furthermore, we define the worst-case regret at step \( h \) as \( \mathrm{reg}_h^t := \max_{i \in [N]} \max_{s \in \mathcal{S}} \mathrm{reg}_h^t(i, s).\) Then, we define the CCE-Gap as the “distance” of a policy to CCE as
\begin{align}
\mathrm{CCE\mbox{-}Gap}(\bar{\pi}) := \max_{i \in [N]} \left[ V_{i,1}^{\dagger,\bar\pi_{-i}}(s_1) - V_{i,1}^{\bar{\pi}}(s_1) \right]
\end{align}

\section{Algorithm and the Main Result}

In this section, we first present Algorithm \ref{alg:markov-dlrc-via-vupdate}, which yields the CCE-gap bound stated in Theorem \ref{thm:dlrc-regret-text}, through the established equivalence between Algorithm \ref{alg:markov-dlrc-via-vupdate} and Algorithm \ref{alg:markov-dlrc-oftrl-lifted} for multi-player general-sum Markov games. Since the algorithms run by all the agents are symmetric, we only illustrate our algorithm using a single agent $i$. Algorithm \ref{alg:markov-dlrc-via-vupdate} consists of three major components: The policy update step that computes the strategy for each matrix game, the value update step that updates the value functions, and the policy output step that generates a CCE policy.

\begin{algorithm}[ht]
  \caption{Markov-Game DLRC-OMWU with V-Updates}
  \label{alg:markov-dlrc-via-vupdate}
  \begin{algorithmic}[1]
    \STATE \textbf{Initialize} $d \!=\! |A_i|$, $U^{(t=1)}_{i,h}(s,\!\cdot),u^{(t=0)}_{i,h}(s,\!\cdot)\! \leftarrow \!\mathbf{0}^d $, and $V^{(0)}_{i,h}(s) \!\leftarrow \!0$ $\forall h \in [H]$, $\forall s\in \mathcal{S}$, $i \in [N]$
    \FOR{$t = 1$ to $T$}
      \STATE \textbf{Policy update:}
      \STATE $\mathcal{R}^{(t)}_{i,h}(s,\cdot) \leftarrow \eta_t\big(U^{(t)}_{i,h}(s,\cdot) + \frac{w_t}{w_{t-1}}u^{(t-1)}_{i,h}(s, \cdot)\big)$
      \STATE $\displaystyle\lambda^{(t)} = \arg\max_{\lambda \in (0,1]} \big[(\alpha-1)\log \lambda + \log \sum_{a'_{i}\in A_i} e^{\lambda\, \mathcal{R}^{(t)}_{i,h}(s,a'_i)} \big]$
      \STATE for all $ a_i \in A_i$ update policies:
        \( 
        \pi^{(t)}_{i,h}(a_i|s) \leftarrow \frac{e^{\lambda^{(t)} \mathcal{R}^{(t)}_{i,h}(s,a_i)}}{\sum_{a'_{i}\in A_i} e^{\lambda^{(t)}\mathcal{R}^{(t)}_{i,h}(s,a'_i)}}
        \) 
     \STATE \textbf{Value update:} for $h = H \to 1$:
        \[
        V^{(t)}_{i,h}(s)
        \leftarrow
        (1 - \alpha_t)\, V^{(t-1)}_{i,h}(s)
        + \alpha_t \big[(r_{i,h} + P_h V^{(t)}_{i,h+1})\, \pi_h^{(t)}\big](s)
        \]
     \STATE Set:
        \( 
        u^{(t)}_{i,h}(s,\cdot)
        \;:=\;
        w_t\Big[\;\big[(r_{i,h} + P_h V^{(t)}_{i,h+1})\;\pi^{(t)}_{-i,h}\big](s,\cdot) -V^{(t)}_{i,h}(s)\mathbf{1}_d\Big]
        \)
      \STATE Update utility vector: $U^{(t+1)}_{i,h}(s,\cdot) = U^{(t)}_{i,h}(s,\cdot) + u^{(t)}_{i,h}(s,\cdot)$
    \ENDFOR
    \RETURN averaged policy $\bar{\pi}$ specified in Algorithm \ref{alg:rollout}
  \end{algorithmic}
\end{algorithm}

\begin{algorithm}[ht]
\caption{Roll-out procedure $\bar\pi_t$ for evaluation}
\label{alg:rollout}
\begin{algorithmic}[1]
  \REQUIRE policy stream $\{\pi^t_h\}_{h,t}$ from Algorithm ~\ref{alg:markov-dlrc-via-vupdate}
  \FOR{$h'=1,\dots,H$}
    \STATE Sample $j \in [T]$ with probability $\mathbb{P}(j=i) = \alpha^j_T$
    \STATE Execute policy $\pi^j_{h'}$ at step $h'$
    \STATE Play policy $\bar\pi^j_{h+1}$ onward
  \ENDFOR
\end{algorithmic}
\end{algorithm}

\subsection{Policy Update.}

At every fixed state-step pair \( (s, \!h) \), the agents engage in a sequence of matrix games, where, in the \( t \)-th iteration, agent \( i \)'s payoff matrix is determined by the estimates of the V-functions. Similarly, it is well known that when all players employ no-regret algorithms in NFGs, their time-averaged joint strategy forms a \(\tfrac{\mathrm{Reg}_T}{T}\)-approximate CCE \cite{cesa2006prediction}. For each state-stage pair, we treat the local interaction as a fixed matrix game and deploy a no-regret algorithm. Thus, we introduce our algorithm, the Markov-Game Dynamic Learning-Rate Control Optimistic Multiplicative Weights Update (MG-DLRC-OMWU). Our proposed algorithm is adapted from the work of \cite{Soleymani_2025}, which reframes equilibrium learning as an learning rate control problem. The underlying principle is a penalization of excessively negative regret, which limits the influence of exceptionally performing actions. This approach is conceptually linked to replicator dynamics in evolutionary game theory \cite{weibull1997evolutionary}, where updates are based on ``harmony''. 

The core of our algorithm is a variant of Optimistic MWU (OMWU) that incorporates an adaptive  regularizer. For each agent $i$, at each fixed $(s,h)$, the algorithm maintains two key components; a cumulative dual vector $  U^{(t)}_{i,h}(s,\cdot) \in \mathbb{R}^{d}$, and an regret correction vector $u^{(t-1)}_{i,h}(s,\cdot) \in \mathbb{R}^d$. In more detail, upon observing the $ \nu^{(t-1)}_{i,h}(s,\cdot) = \big[(r_h + P_h V^{(t-1)}_{h+1,i})\;\pi^{(t-1)}_{-i,h}\big](s,\cdot) \in \mathbb{R}^d$, the algorithm calculates the regret signal:
\( 
u^{(t-1)}_{i,h}(s,\cdot) := w_{t-1} \big(\nu^{(t-1)}_{i,h}(s,\cdot) - \langle \nu^{(t-1)}_{i,h}, \pi_{i,h}^{(t-1)} \rangle(s) \mathbf{1}_d \big).
\) 
Then, these are combined to form an optimistic regret estimate
\( 
\mathcal{R}^{(t)}_{i,h}(s,\cdot) := \eta_t\bigl(U^{(t)}_{i,h}+ \frac{w_t}{w_{t-1}} u^{(t-1)}_{i,h}\bigr)(s,\cdot),
\) in which weights are the value update rates in Algorithm \ref{alg:markov-dlrc-via-vupdate}. Then, the policy $\pi^{(t)}$ is computed as follows, 
\begin{equation}
\pi^{(t)}_{i,h}(a_i|s) := \frac{\exp\left( \lambda^{(t)} \mathcal{R}^{(t)}_{i,h}(s,a_i) \right)}{\sum_{a'_{i}\in A_i} \exp\left( \lambda^{(t)}\mathcal{R}^{(t)}_{i,h}(s,a'_i) \right)}
\end{equation}
A key point in Algorithm \ref{alg:markov-dlrc-via-vupdate} is the dynamic learning-rate control scheme for selecting $\lambda^{(t)}$. This scheme adapts the learning rate based on the magnitude of the optimistic regret. If the rewards are already large, indicating a volatile learning phase between the players, a conservative fixed learning rate $\eta$ is used. Otherwise, the learning rate is optimized to balance the learning progress against the stability of the updates. More clearly, when the parameter $\alpha$ is chosen to be on the order of $\Theta(\log^2 d + \log d)$, it can be shown that $\lambda^{(t)}$ update at Line 5 in Algorithm \ref{alg:markov-dlrc-via-vupdate} is equivalent to:
\begin{align}
\!\lambda^{(t)} = 
\begin{cases}
\eta, & \!\!\!\!\!\text{if } \max_{a'i\in A_i} \!\mathcal{R}^{(t)}_{i,h}(s,a'_i) \!\ge\! -\beta \log \!|\mathcal{A}_{i}|, \\[0.5em]
\displaystyle \arg\max_{\lambda \in (0, 1]} \left\{ (\alpha - 1) \log \lambda + \log \sum_{a'_{i}\in A_i} e^{\lambda\, \mathcal{R}^{(t)}_{i,h}(s,a'_i)}  \right\}\!, & \!\!\!\!\!\text{otherwise} \qquad \qquad \qquad\qquad\qquad
\end{cases}\notag
\end{align}
where $\beta$, $\eta_t$, are hyperparameters chosen accordingly. For the analysis of the Algorithm \ref{alg:markov-dlrc-via-vupdate} we use a weighted time-dependent learning rate schedule within the equivalent OFTRL algorithm, which extends the stationary learning rate analysis in \cite{Soleymani_2025}. This formulation lets players adapt to the non-stationary dynamics while preserving the regret-minimization principles. The analysis of the equivalent algorithm, Algorithm~\ref{alg:markov-dlrc-oftrl-lifted}, lets us to provide the RVU inequality in Theorem~\ref{thm:rvu_tv_text}, as the RVU property is key to achieving sublinear regret and ensuring equilibrium convergence.

\subsection{Value update.}

Each player maintains V value function $V_{i,h}^{(t)}(s)$ for every for every \((h, s)\) pair and conducts smooth value update with the following learning
rates: \( \alpha_t := \frac{H+1}{H+t}, \) proposed by ~\cite{JinAllenZhuBubeckJordan2018}, which guarantees stability across long horizons, and adopted by the wide range of works in the literature \cite{zhang2022policyoptimizationmarkovgames,yang2023ot1convergenceoptimisticfollowtheregularizedleadertwoplayer, cai2024nearoptimalpolicyoptimizationcorrelated, mao2024widetildeot1convergencecoarsecorrelated, JinLiuWangYu2022} Under this step size, the V-update at round \( t \) corresponds to the weighted average:
\begin{align}
V^{(t)}_{i,h}(s)
= \sum_{j=1}^{t}
\alpha_t^j
\cdot
\Big[ (r_{i,h} + P_h V^{(j)}_{i,h+1}) \pi^{(j)}_h \Big](s)
\end{align}
where the time-dependent coefficients \( \alpha_j^t \) are defined as
\( 
\alpha_t^j := \alpha_j \prod_{k=j+1}^t (1 - \alpha_k),\) \text{for } \( j < t,\) \text{and} \(  \alpha_t^t := \alpha_t.
\)
This update ensures that \( \sum_{j=1}^t \alpha_t^j = 1 \), so that the estimate remains a proper average. In this work, we also adopt the same weight sequence \( \{ \alpha_t^j \} \), which is used to construct the utility weights \( w_j \) in the MG-DLRC-OMWU procedure via the relation \( w_j := \alpha_t^j / \alpha_t^1 \) for all \( t \in [T] \).

\subsection{Policy output.}
The final joint policy \( \bar{\pi}^t_h \) is constructed by aggregating the history of policies across time using the same weights \( \alpha_j^t \) that govern the value updates. Formally, at the initial step \( h \in [H] \), we sample an iteration index \( j \in [t] \) with probability proportional to \( \alpha_t^j \), and execute the joint policy \( \pi^{(j)}_h \) at that step. Subsequently, the process continues by executing $\bar\pi^j_{h+1}$ at the next step of the same episode, and proceeds similarly at each following step. This procedure is summarized in Algorithm~\ref{alg:rollout} and to our knowledge it was first proposed in \cite{bai2020nearoptimalreinforcementlearningselfplay}. Since all players sample from the same index \( j \) at each step, the resulting policy \( \bar{\pi} \) is correlated across players similar to \cite{ zhang2022policyoptimizationmarkovgames}. 

\subsection{Analysis}

In this section, we present Algorithm~\ref{alg:markov-dlrc-oftrl-lifted}, which is shown to be equivalent to Algorithm~\ref{alg:markov-dlrc-via-vupdate} in Lemmas~\ref{text_eq_dlrc} and~\ref{lem:q-v-update-equivalence}. More specifically, Algorithm~\ref{alg:markov-dlrc-oftrl-lifted} employs an equivalent formulation of the policy update using an OFTRL mechanism with a specific regularizer in a lifted space, and adopts Q-updates that are equivalent to the V-updates in Algorithm~\ref{alg:markov-dlrc-via-vupdate}. We state the equivalence in policy updates in Lemma \ref{text_eq_dlrc}.

\begin{algorithm}[ht]
 \caption{MG-DLRC-OFTRL}
 \label{alg:markov-dlrc-oftrl-lifted}
 \begin{algorithmic}[1]
  \STATE Initialize: $d\!=\!|A_i|$,  $U^{(t=1)}_{i,h}(s,\!\cdot),u^{(t=0)}_{i,h}(s,\!\cdot)\! \leftarrow \!\mathbf{0}^d$, and $Q^{(0)}_{i,h}(s,\!\cdot) \!\leftarrow \!\mathbf{0}^d$ $\forall h \!\in \![H]$, $\forall s\!\in\! \mathcal{S}$, $i \in [N]$ 
  \STATE Define regularizer: $\Psi^{(t)}(y) = -\tilde\alpha\log\big(\sum_{j=1}^{d}y[j]\big)
   + \frac{1}{\sum_{j=1}^{d}y[j]}\sum_{j=1}^{d}y[j]\log y[j]$
  \FOR{$t = 1$ \TO $T$}
   \STATE \textbf{Policy update:}
   \STATE Optimistic signal: $\mathcal{R}^{(t)}_{i,h}(s,\cdot) \gets \eta_t\big( U^{(t)}_{i,h} + \frac{w_t}{w_{t-1}}u^{(t-1)}_{i,h}\big)(s,\cdot)$
   \STATE OFTRL update: $y_{i,h}^{(t)}(s,\cdot) = \arg\max_{y\in(0,1]\Delta^{d}} \big[\langle \mathcal{R}^{(t)}_{i,h},y\rangle - \Psi^{(t)}(y)\big]$
   \STATE Recover policy for all $ a_i \in A_i$: $\pi^{(t)}_{i,h}(a_i|s)\;=\;\frac{y_{i,h}^{(t)}(s,a_i)}{\sum_{a'_i \in A_i }y_{i,h}^{(t)}(s,a'_i)}$
   \STATE \textbf{Value update:} for $h = H \to 1$:
      \STATE $Q^{(t)}_{i,h}(s,a) = (1-\alpha_t)\,Q^{(t-1)}_{i,h}(s,a)
      + \alpha_t\big[r_{i,h}+  P_h[Q^t_{i,h+1}\pi^t_{h+1}]\big](s,a) $
   \STATE Set: $u^{(t)}_{i,h}(s,\cdot) = w_t\big[\,[Q^{(t)}_{i,h}\,\pi^{(t)}_{h,-i}](s,\cdot)-[Q^{(t)}_{i,h}\,\pi^{(t)}_{h}](s)\mathbf{1}_d\big]$
   \STATE Update utility vector: $ U^{(t+1)}_{i,h}(s,\cdot) = U^{(t)}_{i,h}(s,\cdot) + u^{(t)}_{i,h}(s,\cdot) $
   \ENDFOR
   \RETURN averaged policy $\bar{\pi} $ specified in Algorithm \ref{alg:rollout}
 \end{algorithmic}
\end{algorithm}

\begin{lemma}[Equivalence of DLRC--OMWU Formulations with Time‐Varying Step‐Size] \label{text_eq_dlrc}
Let \(\{ \eta_t \}\) be a sequence of learning rate caps with \(\eta_t \in (0,1]\). For each agent $i$, at each fixed $(s,h) \in \mathcal{S} \times [H]$, define \(  \mathcal{R}^{(t)} := \frac{\eta}{w_t}\Big(U^{(t)} + \frac{w_t}{w_{t-1}}u^{(t-1)}\Big), u^{(t)} := w_t\Big(\nu^{(t)} - \langle \nu^{(t)}, \pi^{(t)}_{i} \rangle \mathbf{1}_d\Big), 
U^{(t)} := \sum_{\tau=1}^{t-1} u^{(\tau)}.
\) First, the Lines 5 and 6 of Algorithm \ref{alg:markov-dlrc-via-vupdate} is an instance of equation \ref{eqeq}. Furthermore, with the variable change \(y^{(t)}=\lambda^{(t)}x^{(t)}\) and  \(\lambda^{(t)}=\sum_k y^{(t)}[k]\) the following two optimization problems are equivalent:
\begin{enumerate}
    \item[1.] \textbf{DLRC--OFTRL in \((\lambda,x)\)-space}
    \begin{equation} \label{eqeq}
    \bigl(\lambda^{(t)},x^{(t)}\bigr)
    \;=\;
    \arg\max_{\substack{\lambda\in(0,1], x\in\Delta^d}}
    \Bigl\{
        \lambda\,\langle  \mathcal{R}^{(t)},x\rangle
        + (\alpha-1)\log\lambda
        - \textstyle\sum_{k=1}^d x[k]\log x[k]
    \Bigr\}
    \end{equation}
    
    \item[2.] \textbf{Lifted optimistic FTRL in \(y\)-space}
    \begin{equation} \label{eqeq2}
    y^{(t)}
    \;=\;
    \arg\max_{y\in(0,1]\Delta^d}
    \Bigl\{
        \,\langle  \mathcal{R}^{(t)},y\rangle
        + \alpha \log\bigl(\textstyle\sum_k y[k]\bigr)
        - \frac{1}{\sum_k y[k]}\sum_{k=1}^d y[k]\log y[k]
    \Bigr\}
    \end{equation}
\end{enumerate}
\end{lemma}

On the other hand, for the value updates we motivate our V-value in Algorithm \ref{alg:markov-dlrc-via-vupdate} as maintaining a state-action value function $Q_{i,h}(s, a)$ requiring storing and updating values for each joint action $a = (a_1, \ldots, a_n)$, leading to a space complexity of $|\mathcal{S}| \cdot \prod_{j=1}^n |\mathcal{A}_j|$. This exponential dependence on the number of agents is often renders $Q$-based methods impractical in scaled multi-agent settings.

To address this issue, we adopt a value-based formulation in which each player $i$ maintains a compact state-value function $V_{i,h}(s)$ that depends only on the state $s$, thus requiring only $\mathcal{O}(|\mathcal{S}|)$ space, following the ideas from \cite{zhang2022policyoptimizationmarkovgames,cai2024nearoptimalpolicyoptimizationcorrelated}. It is also important to note that, this $V$-based approach also enables a decentralized implementation: as the update for $V_{i,h}(s)$ requires only the expected value of $(r_h + P_h V_{i,h+1})$ under the joint policy of the other players, $\pi_{-i,h}(\cdot \mid s)$. Crucially, this expectation can be obtained without explicitly reconstructing $\pi_{-i,h}$, by interacting with the environment or using standard black-box feedback mechanisms. In particular, we assume access to a \emph{reward oracle} that returns $\mathbb{E}_{a_{-i} \sim \pi_{-i,h}}[r_i(s, a_i, a_{-i})]$, and
a \emph{transition oracle} that returns the $ \mathbb{E}_{a_{-i} \sim \pi_{-i,h}}[P_h(\cdot \mid s, a_i,a_{-i})]$ for any fixed $(s, a_i)$, or that both can be approximated through sampling in a model-free setting. As a result, each player can update its value function $V_{i,h}$ using only its own local trajectory data, without needing access to the policies of the other agents. This allows the overall learning process to be implemented in a fully decentralized manner. Following this, we state the equivalence between the $Q$ and $V$-based updates in the Lemma \ref{lem:q-v-update-equivalence}:

\begin{lemma}
\label{lem:q-v-update-equivalence}
The value updates in Algorithm~\ref{alg:markov-dlrc-oftrl-lifted} (the $Q$-update) and Algorithm~\ref{alg:markov-dlrc-via-vupdate} (the $V$-update) are equivalent. Specifically, for all $t \in [T]$, $h \in [H]$, $s \in \mathcal{S}$, and $a_i \in \mathcal{A}_i$, the iterates satisfy
\[
Q_{i,h}^t(s,a_i) = r_h(s,a_i) + \left(P_h V_{i,h+1}^t\right)(s,a_i)
\quad \text{and} \quad
V_{i,h}^t(s) = \left\langle \pi_{i,h}^t(\cdot|s), Q_{i,h}^t(s,\cdot) \right\rangle,
\]
\end{lemma}

\subsection{Learning–Rate Selection and Equivalent Optimistic FTRL View}

By the proven equivalence of Algorithms \ref{alg:markov-dlrc-via-vupdate} and \ref{alg:markov-dlrc-oftrl-lifted}, we now state the CCE convergence of Algorithm~\ref{alg:markov-dlrc-oftrl-lifted} in multi-player general-sum Markov games, which therefore also holds for Algorithm \ref{alg:markov-dlrc-via-vupdate}.

\begin{theorem}[Regret Bounds for MG-DLRC-OMWU]
\label{thm:dlrc-regret-text}
Suppose that $n$ players engage in self-play in a general-sum Markov game with an action set of size bounded by $|\mathcal{A}_{\max}|$, over $T$ rounds. If each player follows Algorithm~\ref{alg:markov-dlrc-oftrl-lifted}, equivalently Algorithm~\ref{alg:markov-dlrc-via-vupdate}, with parameters $\beta \geq 70$, $\tilde{\alpha} = \beta\log^2|\mathcal{A}_{\max}| + 2\log|\mathcal{A}_{\max}|+2$, and $\eta = 1/24H\sqrt{H}N$, then the output policy $\bar\pi$ satisfies: 
    \begin{align}
     \text{CCE-Gap}(\bar\pi) \leq \frac{864H^\frac{7}{2}N\!\bigl(\tilde\alpha\log T+2\log |\mathcal{A}_{\max}|+2\bigr)}{\,T}
    \end{align}
\end{theorem}

Theorem~\ref{thm:dlrc-regret-text} improves the best-known convergence bound to coarse correlated equilibrium (CCE) from $\mathcal{O}((\log T)^5/T)$, as established in~\cite{mao2024widetildeot1convergencecoarsecorrelated}, to $\mathcal{O}(\log T/T)$. This matches the optimal convergence rate for CCE in NFGs, in terms of dependence on both the number of actions and the time horizon. Notably, the analysis in~\cite{mao2024widetildeot1convergencecoarsecorrelated} builds upon the work of~\cite{daskalakis2023nearoptimalnoregretlearninggeneral} on the multiplicative weights update (MWU) algorithm, which introduces a large constant factor on the order of $C\sim10^8$ and guarantees convergence only when $T \geq C|\mathcal{A}_{\max}|H^4$.

In contrast, Theorem~\ref{thm:dlrc-regret-text} closes the gap between the best-known convergence rates for coarse correlated and correlated equilibria. While the existence of algorithms, which converges to CE, is already known to achieve a rate of $\mathcal{O}(\log T/T)$, our result shows that such rates are also attainable for the more general CCE set. This aligns with empirical observations reported in~\cite{mao2024widetildeot1convergencecoarsecorrelated}, suggesting that faster convergence to CCE is indeed achievable in practice. Furthermore, our results improve aforementioned CE convergence bounds in terms of $H$ and $|\mathcal{A}_{\max}|$ terms.  The proof of Theorem~\ref{thm:dlrc-regret-text} follows a similar structure to prior works such as~\cite{mao2024widetildeot1convergencecoarsecorrelated}. For clarity, we begin by introducing a key intermediate quantity. Defining
\(
\delta^t_h \!:= \! \max_{s\in \mathcal{S},i\in[n]} \!\big( V^{\dagger, \bar{\pi}^t_{-i,h}}_{i,h}(s)\! - \!V^{\bar{\pi}^t_h}_{i,h}(s) \big),
\)
we have the following recursive bound:
\begin{lemma}
\label{lem:gap_recursion_text}
For the policy $\bar{\pi}^t_h$, in Algorithm~\ref{alg:rollout}, for all $(i, h, t) \in [n] \times [H] \times [T]$ we have \begin{align} 
\delta^t_h
\le \sum_{j=1}^t \alpha^j_t \delta^t_{h+1} + \max_{s\in \mathcal{S},i\in[n]}\mathrm{reg}^t_{i,h}(s),
\label{recursivegap}
\end{align} 
\end{lemma}
Thus, upper bounding $\text{CCE-Gap}(\bar\pi)$ reduces to controlling the per-state weighted regrets for each player and every $(s, h) \in \mathcal{S} \times [H]$. We next provide a regret bound for each $(i, h, s)$, derived using an RVU-type inequality for the MG-DLRC-OMWU algorithm under time-varying learning rates.
\begin{theorem}[RVU bound for MG-DLRC-OMWU with time-varying $\eta_t$]
\label{thm:rvu_tv_text}
Let $\beta \ge 70$, and \( \kappa^{(t)} = \frac{w_t}{w_{t-1}} \). 
For each $(s,\!h)$, consider the \textnormal{inner} OFTRL process in Algorithm~\ref{alg:markov-dlrc-oftrl-lifted}, with iterates \(y^{(t)}\!, u^{(t)}\!\) as in Lemma~\ref{text_eq_dlrc}. 
Then, the cumulative regret
\( 
\tilde{\operatorname{Reg}}(T)\!
:= \!\!
\sum_{t=1}^{T}\!\!\big\langle y \!-\! y^{(t)}\!,\! u^{(t)} \!\big\rangle
\)
incurred up to horizon $T$ obeys
\begin{align}
\!\!\!\!\!\tilde{\operatorname{Reg}}(T)\!
\le  
\!2\|u^{(t)}\|_\infty\!\!+\!
\frac{\!\tilde\alpha\log T\!\!+\!2\log \!|\mathcal{A}_{\max}|\!}{\eta_{T+1}}
 \!+\!\sum_{t=1}^{T}\!\eta_t\|u^{(t)}\!\!-\!\kappa^{(t)}u^{(t-1)}\|_\infty^2 \!\!- \!
\frac1{20}\!\sum_{t=1}^{T-1}
       \!\frac{\!\|x^{(t+1)}\!-\!x^{(t)}\|_1^{2}\!\!\!}{\eta_t}
\end{align}
\end{theorem}
Now, using the given Theorem \ref{thm:rvu_tv_text} we derive the following  regret bound for each $(i, h, s)$. 
\begin{lemma}[Per-state weighted regret bounds]\label{lem:Per-state_regret}
Fix an episode $h \in [H]$, state $s \in \mathcal{S}$, agent $i \in \mathcal{N}$, and horizon $T \ge 2$. Run Algorithm \ref{alg:markov-dlrc-oftrl-lifted} with a base learning rate $\eta > 0$ and weights $w_j = \frac{\alpha_t^j }{\alpha_1^t}$, Then, 
\begin{align}
\mathrm{reg}_{i,h}^{t}(s)
\le
&
\frac{2H\!\bigl(\tilde\alpha\log t+2\log |A_{\max}|+6H\eta\bigr)}{\eta\,t}
\!+\!
12\eta\,H^{2}(N\!-\!1)\!
      \sum_{j=2}^{t-1}
      \sum_{k\neq i}\!\alpha_t^j
        \bigl\|
          \pi_{h,k}^{j}\!-\!\pi_{h,k}^{j-1}
        \bigr\|_1^{2}
\notag\\
&+
\frac{\!12\eta H^2(3H\!+\!4N^2)\!}{t}
- 
\frac1{24\,\eta H}
      \sum_{j=2}^{t-1}\alpha_{\,t}^{\,j}\,
        \bigl\|
          \pi_{i,h}^{j}-\pi_{i,h}^{j-1}
        \bigr\|_1^{2}.
\end{align}
Moreover, summing over all agents and setting $\eta = 1/(24H\sqrt{H}N)$ yields:
\begin{align}
\!\!\!\sum_{i=1}^{N} \!\mathrm{reg}_{i,h}^{t}(s)
\!\le
\frac{\!\!2HN\!\bigl(\tilde\alpha\!\log \!t \!+ \!2\log\!|\mathcal{A}_{\max}| \!\!+\!\!  6\eta^2 H\!N(3H\!\!+\!4N^2)\!+\!\!6H\eta\bigr)\!\!}{\eta t}
\!-\!
\frac{\!\!\sum_{i=1}^{N}\!
      \sum_{j=2}^{t-1}\!\alpha_t^j\!
        \left\|
          [\pi_{i,h}^{j}-\pi_{i,h}^{j-1}]
        \right\|_1^{2}\!\!}{48\eta H}
\label{eq:new_sum_text}
\end{align}
\end{lemma}

We note a key structural distinction between the recursive bound in~\eqref{recursivegap} and the regret guarantee in~\eqref{eq:new_sum_text}. Specifically, the former requires bounding the \emph{maximum} external regret across agents, while the latter controls only the \emph{sum} of regrets. This distinction is typically problematic since external regret can be negative, and the sum does not necessarily upper bound the maximum. However, in our setting this issue is circumvented by the fact that the RVU bound in Theorem~\ref{thm:rvu_tv_text} applies to policy iterates $y^{(t)}$ of the OFTRL process, which are guaranteed to yield non-negative regrets due to the structure of the MG-DLRC-OFTRL update. This non-negativity, formally established in Proposition~\ref{prop:nonnegative_regret} in Appendix \ref{app:b}, stands in sharp contrast to general external regret and allows us to upper bound the maximum regret by summing the external regret bound over all players. Thus, using the RVU bound we can bound the second order path length and derive the final bound as stated in Theorem \ref{thm:dlrc-regret-text}, which has been stated in Appendix C. This completes the high-level overview of our analysis.

\section{Numerical Results}

\begin{figure}
\centering
\subfigure[Trajectory of the average CCE-gap, over 9 games]{\label{fig:a}\includegraphics[width=80mm]{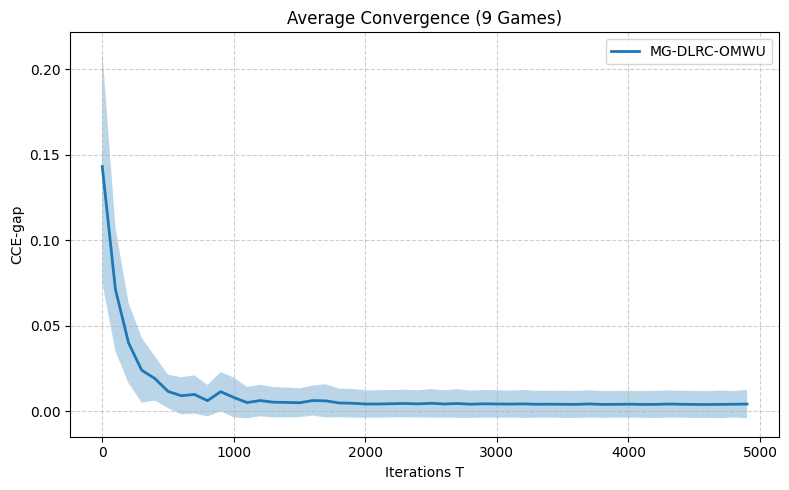}}
\hspace{10mm}
\subfigure[CCE-gap across 9 games]{\label{fig:b}\includegraphics[width=60mm]{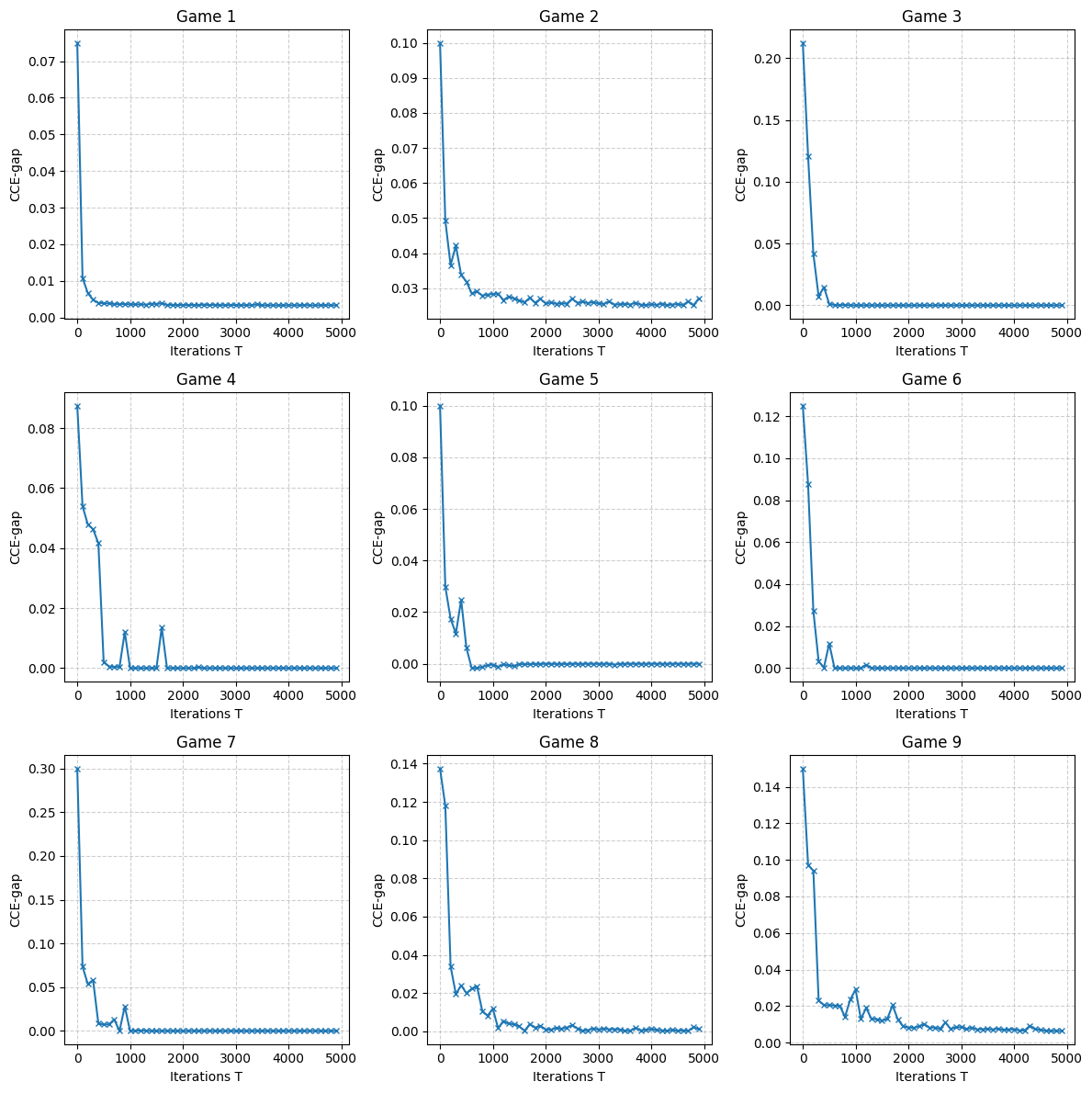}}
\label{fig:overall}
\end{figure}

In this section, we numerically evaluate our proposed algorithm MG-DLRC-OMWU on a set of general-sum Markov games. Each environment consists of 2 players, 2 states, and 2 actions per player, with a horizon length of $H = 2$. The rewards are generated within the interval of $[0,1]$ for each trial, while the transitions are fixed: the system stays in the current state with probability~0.8 and transitions to the other state with probability~0.2. Figure~\ref{fig:overall}-(a) reports the trajectory of the average CCE-gap over 9 independent simulations. In that figure, we show the mean convergence trajectory across 9 games, and the shaded region represents one standard deviation around the mean. In each chosen game, we observe that the CCE-gap converges with a rate of $\mathcal{O}(\log T / T)$. Finally, we show the individual trajectories across the 9 game instances in Figure~\ref{fig:overall}-(b).

\section{Conclusion}

In this work, we have introduced a policy optimization algorithm that achieves a near-optimal convergence rate of $\mathcal{O}(\log T / T)$ to the CCE in general-sum Markov games. This result improves upon the previous best-known convergence rate of $\mathcal{O}((\log T)^5 / T)$ for CCE learning in such settings, while matching the fastest convergence rate for the CE in general-sum Markov games, both established by \cite{mao2024widetildeot1convergencecoarsecorrelated}. While achieving constant regret remains a challenging open problem in general-sum Markov games, recent advances in the zero-sum setting, \cite{yang2023ot1convergenceoptimisticfollowtheregularizedleadertwoplayer}, provide a promising foundation. In particular, ideas inspired by social computation theory, as explored in \cite{Soleymani_2025}, have shown that promoting coordination among agents can lead to improved convergence guarantees. Such techniques may lend themselves to bridging the current gap and achieving constant regret in general-sum settings as well. Moreover, future directions include improving the convergence rates by enhancing the algorithm’s dependence on the horizon $H$ and the maximal action space size $|A_{\max}|$. Another important extension would be moving from the oracle setting, where reward and transition models are assumed to be known, to a more realistic sample-based setting where the game parameters must be learned through interaction.

\acks{Research of the authors was supported in part by the Army Research Office (ARO) Grant Number W911NF-24-1-0085.}

\bibliography{references}

\newpage
\appendix 
\section{Technical Tools and Their Proofs}

\noindent\textbf{Lemma \ref{lem:q-v-update-equivalence}} \textbf{(Equivalence of V and Q-Updates)}

{\itshape Algorithms \ref{alg:markov-dlrc-oftrl-lifted} (Q‐based) and \ref{alg:markov-dlrc-via-vupdate} (V‐based) generate the identical policy sequence $\{\pi^t_{h}\}(s,\cdot)$.  Equivalently, for all agents $i$, states $s$, steps $h$, actions $a$, and rounds $t$, we have}
\begin{equation}\label{eq:Q‐V‐identity}
  Q^t_{i,h}(s,a)
  \;=\;
  r_{i,h}(s,a)\;+\;\bigl[P_h\,V^t_{i,h+1}\bigr](s,a).
\end{equation}

\begin{proof}
The base case $t=0$ holds by initialization, and for $t=1$, as we have $\alpha_1 =1$, equality holds by definition. Furthermore, it can also be seen that we have $Q^t_{i,H}(s,a) = r_{i,H}(s,a)$. Then, suppose that \eqref{eq:Q‐V‐identity} holds for all rounds up to $t-1$ and all levels $\ge h+1$.  Then, writing the Q‐value update from Algorithm \ref{alg:markov-dlrc-oftrl-lifted},
\[
  Q^t_{i,h}(s,a)
  \;=\;
  (1-\alpha_t)\,Q^{t-1}_{i,h}(s,a)
  \;+\;\alpha_t\Bigl(r_{i,h}(s,a)+P_h\,[Q^t_{i,h+1}\,\pi^t_{h+1}](s,a)\Bigr).
\]
By induction on $(t-1,h)$ and on $(t,h+1)$, each occurrence of $Q^t_{i,h+1}$ is replaced by
\( 
  Q^t_{i,h+1}(s,a)
  = r_{i,h+1}(s,a) + [P_{h+1}\,V^t_{i,h+2}](s,a),
\)
and thus we have; 
\begin{align*}
    Q^t_{i,h}(s,a)
  & = (1-\alpha_t)[r_{i,h}+P_hV^{t-1}_{i,h+1}](s,a) + \alpha_t (r_{i,h}+P_h[(r_{i,h+1}+P_{h+1}V^t_{i,h+2})\pi^t_{h+1}]) \\  & 
  = \Big[r_{i,h} + P_h \Big((1-\alpha_t)V^{t-1}_{i,h+1}+\alpha_t\big[\big(r_{i,h+1}+P_{h+1}(V^t_{i,h+2})\big)\pi^t_{h+1}\big]\Big)\Big] \\ & = r_{i,h}(s,a)\;+\;\bigl[P_h\,V^t_{i,h+1}\bigr](s,a),
\end{align*}
where, the final step is due to the update rule on the V-values in Algorithm \ref{alg:markov-dlrc-via-vupdate}. This closes the inductive step. Hence the proof is complete. 
\end{proof}
\noindent Now, we establish following the two lemmas, for the recursive regret bounds. 
\begin{lemma}[Equivalence of value functions] \label{lem:value_equivalence}
For Algorithm \ref{alg:rollout}, we have for all players $i \in [m]$ and all $(h, s, t) \in [H + 1] \times S \times [T]$, that, 
$ V^t_{i,h}(s) = V^{\bar{\pi}^t_h}_{i,h}(s). $
\end{lemma}
\begin{proof}
We prove this by backward induction on $h \in [H + 1]$. The claim trivially holds for the base case $h = H + 1$, since all values are zero. Now, suppose that the claim holds for step $h+1$ and all $(s, t) \in S \times [T]$. For step $h$ and any fixed $(s, t)$, we have:
\begin{align*}
    V^t_{i,h}(s) &= \sum_{j=1}^t \alpha^j_t \left\langle Q^j_{i,h}, \pi^j_h \right\rangle(s) \tag{i} \\
    &= \sum_{j=1}^t \alpha^j_t \left\langle r_h + P_hV^j_{i,h+1}, \pi^j_h \right\rangle(s) \\
    &= \sum_{j=1}^t \alpha^j_t \left\langle r_h + P_hV^{\bar{\pi}^j_{h+1}}_{i,h+1}, \pi^j_h \right\rangle(s) \tag{ii} \\
    &= V^{\bar{\pi}^t_h}_{i,h}(s), \tag{iii}
\end{align*}
where the first and third steps follow from definition and the second step is due to the induction step. This proves the claim for step $h$ and thus completes the proof by induction.
\end{proof}

\noindent\textbf{Lemma \ref{lem:gap_recursion_text}} 
{\itshape For the policy $\bar{\pi}^t_h$ defined in Algorithm \ref{alg:rollout}, we have, for all $(i, h, t) \in [n] \times [H] \times [T],$ that the CCE gap is bounded recursively:}
\begin{align*}
\max_{s\in \mathcal{S},i\in[n]} \left[ V^{\dagger, \bar{\pi}^t_{-i,h}}_{i,h}(s) - V^{\bar{\pi}^t_h}_{i,h}(s) \right]
\le \sum_{j=1}^t \alpha^j_t \max_{s'\in \mathcal{S},i\in[n]} \left[ V^{\dagger, \bar{\pi}^j_{-i,h+1}}_{i,h+1}(s') - V^{\bar{\pi}^j_{h+1}}_{i,h+1}(s') \right] + \max_{s\in \mathcal{S},i\in[n]}\mathrm{reg}^t_{i,h}(s).
\end{align*}

\begin{proof}
Fix $(i, h, t) \in [n] \times [H] \times [T]$. We have, for all states $s \in S$, that
\begin{align*}
    &V^{\dagger, \bar{\pi}^t_{-i,h}}_{i,h}(s)\! - \!V^{\bar{\pi}^t_h}_{i,h}(s) = \max_{\pi^{\dagger}_i} \sum_{j=1}^t \alpha^j_t \mathbb{E}_{\pi^{\dagger}_i \times \pi^j_{-i,h}} \!\!\left[ r_h\! +\! P_h V^{\dagger, \bar{\pi}^j_{-i,h+1}}_{i,h+1}\! \right]\!(s)\! -\!\! \sum_{j=1}^t \!\alpha^j_t \mathbb{E}_{\pi^j_h} \!\left[ r_h\! +\! P_h V^{\bar{\pi}^j_{h+1}}_{i,h+1} \!\right]\!(s) \\
    &\le \!\sum_{j=1}^t\! \alpha^j_t \!\max_{s'\in S} \left[ \!V^{\dagger, \bar{\pi}^j_{-i,h+1}}_{i,h+1}(s') \!-\! V^{\bar{\pi}^j_{h+1}}_{i,h+1}(s') \!\right] \! + \!\max_{\pi^{\dagger}_i} \!\sum_{j=1}^t \!\alpha^j_t\! \left\langle \!\pi^{\dagger}_i(\cdot|s) \!-\! \pi^j_{i,h}(\cdot|s),\! \left[\! (r_h \!+\! P_hV^{\bar{\pi}^j_{h+1}}_{i,h+1})\pi^j_{-i,h} \right]\!(s, \cdot) \!\right\rangle \\
    &= \sum_{j=1}^t \alpha^j_t \max_{s'\in S} \left[ V^{\dagger, \bar{\pi}^j_{-i,h+1}}_{i,h+1}(s') - V^{\bar{\pi}^j_{h+1}}_{i,h+1}(s') \right] + \underbrace{\max_{\pi^{\dagger}_i} \sum_{j=1}^t \alpha^j_t \left\langle \pi^{\dagger}_i - \pi^j_{i,h}, \left[ (r_h + P_hV^{j}_{i,h+1})\pi^j_{-i,h} \right] \right\rangle}_{\text{reg}^t_{i,h}(s)} \\
    &\le \sum_{j=1}^t \alpha^j_t \max_{s'\in S} \left[ V^{\dagger, \bar{\pi}^j_{-i,h+1}}_{i,h+1}(s') - V^{\bar{\pi}^j_{h+1}}_{i,h+1}(s') \right] + \text{reg}^t_{i,h}(s).
\end{align*}
Taking $\max_{s\in \mathcal{S},i\in[n]}$ of both sides concludes the proof. 
\end{proof}
We next present some basic algebraic properties of the weights $\alpha_t = \frac{H+1}{H+t}$, $\{\alpha^i_t\}_{t \ge 1,\ 1 \le i \le t}$ and $\{w_t\}_{t \ge 1}$, which will be used in later proofs. We define them as: 
\begin{align*}
\alpha^t_t = \alpha_t, \quad
\alpha^i_t = \alpha_i \prod_{j = i+1}^{t} (1 - \alpha_j), \quad w_t = \frac{\alpha^t_t}{\alpha^1_t}, \quad \forall\, i \le t - 1. 
\end{align*}

\begin{lemma}\label{lem:A1}
Let \(H\ge1\), and for each \(t\ge1\) let \(\alpha_t>0\) be the step–size above. Then, for every integer \(T\ge1\), the following properties hold:
\begin{enumerate}
  \item \(\displaystyle \sum_{j=1}^T \alpha_T^j = 1.\)
  \item The sequence \(j\mapsto\alpha_T^j\) is non‑decreasing in \(j\).
  \item \(\displaystyle \sum_{j=1}^T (\alpha_T^j)^2 
        \;\le\;\sum_{j=1}^T \alpha_j^2 \;\le\; H+2.\)
  \item For any non‑increasing sequence \(\{b_j\}_{j=1}^T\),
  \( 
    \sum_{j=1}^T \alpha_T^j\,b_j
    \;\le\;
    \frac1T\sum_{j=1}^T b_j.
  \)        
  \item \( \alpha^1_t \le \frac{1}{t}\) 
  \item \(\displaystyle \sum_{j=1}^T\frac{\alpha_T^j}{j}
        \;\le\;\frac{1+\tfrac1H}{T}.\)
  \item \(\displaystyle \sum_{j=1}^T\alpha_T^j\,\alpha_j^2
        \;\le\;\frac{3H}{T}.\)

\end{enumerate}
\end{lemma}
\begin{proof}
Proof of the first five properties can be found in Lemma 6 of \cite{yang2023ot1convergenceoptimisticfollowtheregularizedleadertwoplayer}. We prove the fifth property as follows. First recall the closed form
\[
\alpha_t^j
=\;
(H+1)\,
\frac{(t-1)!\,(H+j-1)!}{(j-1)!\,(H+t)!},
\qquad 1\le j\le t.
\]
Hence
\begin{align*}
\sum_{j=1}^{t}\frac{\alpha_t^{\,j}}{j}
&=
(H+1)\,\frac{(t-1)!}{(H+t)!}
\sum_{j=1}^{t}\frac{(H+j-1)!}{(j-1)!\,j}\\[4pt]
&=
(H+1)\,\frac{(t-1)!}{(H+t)!}(H-1)!
\sum_{j=1}^{t}\binom{H+j-1}{j}\\[4pt]
&=
(H+1)\,\frac{(t-1)!\,(H-1)!}{(H+t)!}
\Bigg[\binom{H+t}{t}-1\Bigg]\\[4pt]
&=
\frac{H+1}{H\,t}
\;-\;
(H+1)\,\frac{(t-1)!\,(H-1)!}{(H+t)!}
\;\le\;
\frac{1+\tfrac1H}{t}.
\end{align*}
where we used “hockey-stick” identity in the third step. For the last property, we know that the sequence \(\{\alpha_j^2\}\) is non‑increasing. Using the fourth property, with 
\(b_j=\alpha_j^2\), along with the third property, yields
\[
\sum_{j=1}^{t}\alpha_t^{\,j}\alpha_j^2
\;\le\;
\frac1t\sum_{j=1}^{t}\alpha_j^2
\;\stackrel{}{\le}\;
\frac{H+2}{t}.
\]
For all \(H\ge1\), one has \(H+2\le3H\), and thus
\( 
{\;
\sum_{j=1}^{t}\alpha_t^{\,j}\alpha_j^2
\;\le\;
\frac{3H}{t}\;}.
\) 
\end{proof}

\begin{lemma}[Pinsker's Inequality]\label{pinsker}
For discrete distributions \(p,q\) on support size \(d\), we have the
\( 
\|p - q\|_{1}^{2} \;\le\; 2\,\mathrm{KL}(p\|q).
\) 
\end{lemma}

\begin{lemma}[Entropy difference]\label{entropydifference}
For discrete random variables \(p,q\) on support size \(d\), we have
\( 
\bigl|H(p) - H(q)\bigr|
\;\le\;
(\log d)\,\sqrt{2\,\mathrm{KL}(p\|q)}.
\) 
\end{lemma}

\section{Proof of RVU bound with time-varying learning rates}\label{app:b}
In this appendix, we present time-varying analogues of the lemmas used to derive the RVU bound in \cite{Soleymani_2025}, to ensure completeness. For a more detailed exposition and the corresponding results under a constant step-size $\eta$, we refer the reader to \cite{Soleymani_2025}. Throughout the rounds \(t \in [T]\), we allow the learning-rate cap to vary, writing \(\eta_t \in(0,1]\). Also, as we use the mentioned RVU bound to upper bound regret for all $(s,h)$ pairs, we do not explicitly denote them throughout the appendix. Furthermore, we let $|\mathcal{A}_{i}| = d$. The optimistic FTRL step in lifted coordinates therefore uses the regularizer
\begin{equation}
\label{eq:def_psit}
{\;
\psi(y)
\;:=\;
-\tilde\alpha \,
      \log\Bigl(\Lambda(y)\Bigr)
\;+\;
\frac{1}{\Lambda(y)}
      \sum_{k=1}^{d} y[k]\log y[k]},
\qquad
\Lambda(y):=\sum_{k=1}^{d}y[k],
\end{equation}
where \(\tilde\alpha =\beta\log^2d+2\log d+2 \) and the update rule
\begin{align}
y^{(t)} \;=\;
\arg\max_{y\in(0,1]\Delta^{d}}
\left\{\,
     \bigl \langle r^{(t)},\,y\bigr\rangle - \psi(y)
\right\}.
\end{align}
Also, we define
\begin{align}
\mathcal{R}^{(t)} := \frac{\eta}{w_t}\Big(U^{(t)} + \frac{w_t}{w_{t-1}}u^{(t-1)}\Big), \quad
u^{(t)} := w_t\Big(\nu^{(t)} - \langle \nu^{(t)}, x^{(t)} \rangle \mathbf{1}_d\Big), \quad
U^{(t)} := \sum_{\tau=1}^{t-1} u^{(\tau)}.
\end{align}
Furthermore, we have $\|\nu^{(t)}\|_{\infty} \leq H$, due to $\nu^{(t)} =Q^t_{i,h}\pi^t_{-i}$, and $x^{(t)}=\pi_{i,h}^{(t)}$ due to Algorithm \ref{alg:markov-dlrc-via-vupdate}. Finally, we define sets, \( 
\Delta_d := \bigl\{\, x \in \mathbb{R}_{\ge 0}^d : \langle \mathbf{1}, x \rangle = 1 \,\bigr\},\) and \( [0,1]\Delta_d := \bigl\{\, y \in \mathbb{R}_{\ge 0}^d : \langle \mathbf{1}, y \rangle \le 1 \,\bigr\}.\) 

Since the RVU bound is derived using aspects from both  Algorithm \ref{alg:markov-dlrc-via-vupdate} and ~\ref{alg:markov-dlrc-oftrl-lifted}, we first establish the equivalence of the policy update steps 
between Algorithms~\ref{alg:markov-dlrc-via-vupdate} and \ref{alg:markov-dlrc-oftrl-lifted}. 
To this end, we prove Lemma~\ref{text_eq_dlrc} in 
Lemma~\ref{lem:dlrc_equivalence} and Corollary~\ref{cor:softmax_lambda_opt}, 
which formally demonstrate the equivalence of their policy update procedures.

\begin{lemma}[Equivalence of DLRC-OMWU Formulations with Time‐Varying Step‐Size]
\label{lem:dlrc_equivalence}
The following two optimization problems are equivalent:

\begin{enumerate}
    \item[1.] \textbf{DLRC--OFTRL in \((\lambda,x)\)-space}
    \begin{equation}
    \label{eq:dlrc_oftrl_lambda}
    \bigl(\lambda^{(t)},x^{(t)}\bigr)
    \;=\;
    \arg\max_{\substack{\lambda\in(0,1]\\ x\in\Delta^d}}
    \Bigl\{
        \lambda\langle \mathcal{R}^{(t)},x\rangle
        + (\tilde\alpha -1)\log\lambda
        - \textstyle\sum_{k=1}^d x[k]\log x[k]
    \Bigr\}.
    \end{equation}
    
    \item[2.] \textbf{Lifted optimistic FTRL in \(y\)-space}
    \begin{equation}
    \label{eq:dlrc_oftrl_y}
    y^{(t)}
    \;=\;
    \arg\max_{y\in(0,1]\Delta^d}
    \Bigl\{
        \,\langle \mathcal{R}^{(t)},y\rangle
        + \tilde\alpha  \log\!\bigl(\textstyle\sum_k y[k]\bigr)
        - \frac{1}{\sum_k y[k]}\sum_{k=1}^d y[k]\log y[k]
    \Bigr\},
    \end{equation}
    with the variable change \(y^{(t)}=\lambda^{(t)}x^{(t)}\) and
    \(\lambda^{(t)}=\sum_k y^{(t)}[k]\).
\end{enumerate}
\begin{proof}
Let \(y=\lambda x\).  Then, \(\sum_{k}y[k]\!=\!\lambda\) and
\( 
y\!\in\!(0,1]\Delta^d
\!\Longleftrightarrow
\!\lambda\!\in\!(0,1],x\!\in\!\Delta^d
\). By direct algebra,
\begin{align*}
\langle \mathcal{R}^{(t)},y\rangle
  \!+\!\tilde\alpha \log\bigl(\textstyle\sum_{k}y[k]\bigr)
  \!-\!\frac{1}{\sum_{k}y[k]}\sum_{k}y[k]\log y[k]
&=\!\langle \mathcal{R}^{(t)},\lambda x\rangle
  \!+\!\tilde\alpha \!\log\lambda
  \!-\!\frac{1}{\lambda}\!\sum_{k}(\lambda x[k])\log(\lambda x[k])\\
&=\lambda\langle \mathcal{R}^{(t)},x\rangle
 \! +\!\tilde\alpha\! \log\lambda
  \!-\!\sum_{k}x[k]\bigl(\log\!\lambda\!+\!\log x[k]\bigr)\\
&=\lambda\langle \mathcal{R}^{(t)},x\rangle
  \!+\!(\tilde\alpha\! -\!1)\log\lambda
  \!-\!\!\sum_{k}x[k]\log x[k].
\end{align*}
Thus, under the bijection \(y=\lambda x\), the two problems are equivalent.
\end{proof}

\end{lemma}

\begin{corollary}[Softmax Structure and Learning Rate Maximization]
\label{cor:softmax_lambda_opt}
Let \((\lambda^{(t)}, x^{(t)})\) be the solution to the DLRC–OFTRL problem in Lemma~\ref{lem:dlrc_equivalence}. Then:
\begin{enumerate}
    \item[(i)] The policy \(x^{(t)}\) is given by a softmax:
    \begin{equation} \label{eq:softmax_update}
    x^{(t)}[k] = \frac{\exp\left( \lambda^{(t)} \mathcal{R}^{(t)}[k] \right)}{\sum_{j=1}^d \exp\left( \lambda^{(t)} \mathcal{R}^{(t)}[j] \right)}.
    \end{equation}
    
    \item[(ii)] The learning rate \(\lambda^{(t)}\) is the solution to the following univariate maximization:
    \begin{equation} \label{univariate}
    \lambda^{(t)} = \arg\max_{\lambda \in (0, 1]} 
    \left\{
    \log \left( \sum_{k=1}^d \exp\left( \lambda \mathcal{R}^{(t)}[k] \right) \right)
    + (\tilde\alpha  - 1) \log \lambda
    \right\}.
    \end{equation}
\end{enumerate}
\end{corollary}

\begin{proof}
\begin{enumerate}
        \item[(i)] Fix \(t\) and write down the Lagrangian of~\eqref{eq:dlrc_oftrl_lambda} with multiplier
    \(\mu\) for the simplex constraint \(\sum_k x[k]=1\) as:
    \[
    \mathcal{L}(\lambda,x,\mu)
    =
    \lambda\langle \mathcal{R}^{(t)},x\rangle
    +(\tilde\alpha -1)\log\lambda
    -\sum_k x[k]\log x[k]
    +\mu\Bigl(1-\textstyle\sum_k x[k]\Bigr).
    \]
    KKT stationarity in \(x\) gives
    \(
    \lambda^{(t)} \mathcal{R}^{(t)}[k]-\log x^{(t)}[k]-1-\mu=0
    \),
    and thus \(x^{(t)}[k]\propto\exp\bigl(\lambda^{(t)} \mathcal{R}^{(t)}[k]\bigr)\).
    Normalizing over \(k\) yields the soft-max form~\eqref{eq:softmax_update}.
    
        \item[(ii)] Given \(\lambda^{(t)}\) from~\eqref{eq:dlrc_oftrl_lambda},
    plugging the soft-max expression~\eqref{eq:softmax_update} back into the
    objective of~\eqref{eq:dlrc_oftrl_lambda} and maximizing over \(\lambda\)
    recovers exactly the univariate problem \eqref{univariate},
    which is the dynamic learning-rate rule of MG-DLRC-OMWU.
    This completes the proof.
\end{enumerate}
\end{proof}

Now, as the equivalence between policy update steps of Algorithms \ref{alg:markov-dlrc-via-vupdate} and \ref{alg:markov-dlrc-oftrl-lifted} has proven, we proceed with the analysis.

\begin{theorem}[Sensitivity of learning rates on regrets]\label{thm:sensitivity}
There exists a universal constant $\beta \geq 70$ such that for \( \eta = \frac{1}{24H\sqrt{H}N}\), $\tilde\alpha \geq 2 + 2 \log d + \beta \log^2 d$, the following property holds. Let $\mathcal{R}, \mathcal{R}' \in \mathbb{R}^d$ be such that $\|\mathcal{R} - \mathcal{R}'\|_\infty \leq 2H\eta$, and let $\hat\lambda, \hat\lambda'$ be the corresponding learning rates defined as
\[
\hat\lambda = \arg\max_{t \in (0, 1]} f(t; \mathcal{R}), \qquad
\hat\lambda' = \arg\max_{t \in (0, 1]} f(t; \mathcal{R}'),
\]
where the function $f$ is given by
\( 
f(\lambda; \mathcal{R}) := (\tilde\alpha  - 1) \log \lambda + \log \left( \sum_{k=1}^d e^{\lambda \mathcal{R}[k]} \right).
\) 
Then, $\hat\lambda$ and $\hat\lambda'$ are multiplicatively stable; specifically,
\[
\frac{7}{10} \leq \frac{\hat\lambda}{\hat\lambda'} \leq \frac{7}{5}.
\]
\end{theorem}
\begin{proof}
The result follows directly by extending Theorem~3.5 and the preceding lemmas from \cite{Soleymani_2025}, where the reward signals satisfy the uniform bound $\|\nu^{(t)}\|_{\infty} \leq 1$, to our case with the bounds $\|\nu^{(t)}\|_{\infty} \leq H$, and step size $w_t$.
\end{proof}

\vspace{.5em}
\begin{theorem}[Strong convexity of the time–varying regularizer]
\label{thm:psi_curvature_timevary}
Fix \(d\ge2\) and set
\(
\alpha = 2 + 2\log d + \beta\log^{2}d
\)
with \(\beta\ge 70\).
For every round \(t\) and every \(y\in(0,1]\Delta^{d}\), the Hessian
of~\eqref{eq:def_psit} satisfies
\begin{equation} \label{eq:Hessian-lowerbound}
\nabla^{2}\psi(y)
\;\succeq\;
\frac{1}{2}
\operatorname{diag}\Bigl(
  \tfrac{1}{y[1]\,\Lambda(y)},\dots,\tfrac{1}{y[d]\,\Lambda(y)}
\Bigr),
\end{equation}
\end{theorem}

\begin{proof}
Write \(x[k]=y[k]/\Lambda(y)\in\Delta^{d}\).
The first-order partial derivative of~\eqref{eq:def_psit} is
\[
\frac{\partial\psi}{\partial y[i]}
=
-\frac{\tilde\alpha }{\Lambda(y)}
-\frac{1}{\Lambda(y)^{2}}\sum_{k}y[k]\log y[k]
+\frac{1+\log y[i]}{\Lambda(y)}.
\]
Differentiating again gives, for every \(i,j\in[d]\),
\[
\frac{\partial^{2}\psi}{\partial y[i]\partial y[j]}
=
\frac{\tilde\alpha -2 + 2\sum_{k}x[k]\log x[k]}{\Lambda(y)^{2}}
\;-\;
\frac{\log x[i]+\log x[j]}{\Lambda(y)^{2}}
\;+\;
\frac{\mathbf 1_{i=j}}{y[i]\Lambda(y)}.
\]
Substitute \(\tilde\alpha =2+2\log d+\alpha'\) with \(\alpha'\ge2\log^{2}d\).
For any vector \(v\in\mathbb{R}^{d}\),
\[
\Lambda(y)^{2}\,v^{\top}\!\nabla^{2}\psi(y)\,v
\;\ge\;
\alpha'\bigl(\textstyle\sum_{k}v[k]\bigr)^{2}
\;+\;
\sum_{k}\frac{v[k]^{2}}{x[i]}
\;-\;
2\Bigl(\sum_{k}v[k]\log x[i]\Bigr)\!\Bigl(\sum_{j}v[j]\Bigr).
\]
The final mixed term is controlled by
\( -2(\sum_k v[k]\log x[i]) (\sum_k v[k])\ge -2\log^{2}d\,(\sum_k v[k])^{2} - \sum_{k}\frac{v[k]^{2}}{2x[i]}\),
and this is exactly absorbed by the choice \(\alpha'\ge2\log^{2}d\).
Hence,
\(
\Lambda(y)^{2}\,v^{\top}\nabla^{2}\psi(y)v
\ge
\sum_{k} \tfrac{v[k]^{2}}{2x[i]},
\)
which is equivalent to~\eqref{eq:Hessian-lowerbound}, which concludes the proof.
\end{proof}
\noindent Now, for the rest of this appendix section, we introduce the following notation. For any \(y,z\in(0,1]\Delta^{d}\), we define
\( 
x[k]=\frac{y[k]}{\Lambda(y)},
\) \( 
\theta[k]=\frac{z[k]}{\Lambda(z)},
\) \( 
\rho:=\frac{\Lambda(z)}{\Lambda(y)}.
\)

\begin{proposition}[Decomposition of the Time--Varying Bregman Divergence]
\label{prop:timevary_bregman}
The Bregman divergence induced by \(\psi\) satisfies
\begin{equation} \label{eq:Psi-Bregman}
{\;
D_{\psi}(z\,\|\,y)
=
(\tilde\alpha -1)\,D_{\log}(\Lambda(z)\|\Lambda(y))
\;+\;
\rho\,\text{KL}\bigl(\theta\,\|\,x\bigr)
\;+\;
(1-\rho)\,\bigl[\,H(\theta)-H(x)\bigr]\;},
\end{equation}
where \(D_{\log}(u\|v)=\log\!\frac{v}{u}+\frac{u}{v}-1\) is the log–regularizer divergence,
\(\text{KL}(\theta\|x)=\sum_k \theta[k]\log\frac{\theta[k]}{x[k]}\) is the Kullback–Leibler
divergence, and \(H(x)=-\sum_k x[k]\log x[k]\) is the entropy.
\end{proposition}

\begin{proof}
Write the gradient of \(\psi\):
\[
\frac{\partial\psi}{\partial y[i]}
=
-\frac{\tilde\alpha -1}{\Lambda(y)}
-\frac{1}{\Lambda(y)^{2}}\sum_{k}y[k]\log y[k]
+\frac{\log y[i]}{\Lambda(y)}.
\]
Using \(x[i]=y[i]/\Lambda(y)\) and \(\Lambda(y)>0\),
\[
\frac{\partial\psi}{\partial y[i]}
=
-\frac{\tilde\alpha -1}{\Lambda(y)}
-\frac{1}{\Lambda(y)}\sum_{k}x[k]\log x[k]
+\frac{\log x[i]}{\Lambda(y)}.
\]
By definition,
\[
D_{\psi}(z\|y)
=
\bigl[\psi(z)-\psi(y)\bigr]
-\sum_{i}\frac{\partial\psi}{\partial y[i]}\,(z[i]-y[i]).
\]
Insert the explicit forms of \(\psi\) and its gradient, factor out
\(\Lambda(y)\), and rearrange terms; after straightforward algebra one obtains
\[
D_{\psi}(z\|y)
\!= \!
(\tilde\alpha -1)(\rho-1)\!-\!(\tilde\alpha \!-\!1)\!\log\rho 
\!+\!(\rho-1)\!\sum_k\! x[k]\log x[k]
\!+\!\sum_k\!\theta[k]\!\log\theta[k]
\!-\!\rho\sum_k\!\theta[k]\!\log \!x[k].
\]
Then, adding and subtracting \( \rho\sum_k\theta[k]\log\theta[k]\), and grouping the terms accordingly we get: 
\[
D_{\psi}(z\|y)
=
(\tilde\alpha -1)D_{\log}(\Lambda(z)\|\Lambda(y)) + \frac{\Lambda(z)}{\Lambda(y)} \text{KL}(\theta\|x)
+(1-\rho)(H(\theta)-H(x))
\]
where, \(D_{\log}\) is the Bregman divergence induced by $-\log(x)$ function due to its strict convexity. Thus, the proof is concluded.
\end{proof}

\begin{proposition}[Strong convexity on the lifted simplex]
\label{prop:bregman_lifted}
For all \(y,z\in(0,1]\Delta^d\),
\begin{align}
D_{\psi}(y\,\|\,z)\;\;\ge\;\;
\frac{1}{2}\,\|y-z\|_1^{\,2}.
\end{align}
\vspace{-0.3em}
\end{proposition}

\begin{proof}
By Theorem~\ref{thm:psi_curvature_timevary}, for any $\nu\in\mathbb{R}^d$,
\begin{align*}
\nu^\top\nabla^2\psi(y)\,\nu
& \;\ge\; 
\frac{1}{2}\sum_{i=1}^d\frac{\nu_i^2}{y_i\,\Lambda(y)} \\ & \;\ge\; \frac{1}{2}\sum_{i=1}^d\frac{\nu_i^2}{y_i} \ge \frac{\Lambda(y)}{2}\sum_{i=1}^d\frac{\nu_i^2}{y_i}
\end{align*}
Since $0<\Lambda(y)\le1$, and by Cauchy–Schwarz we get, 
\begin{align*}
\frac{\Lambda(y)}{2}\sum_{i=1}^d\frac{\nu_i^2}{y_i}
& \;\ge\; 
\frac{1}{2}\sum_{i=1}^d\nu_i^2 \\ & \;\ge\; \frac{1}{2}\sum_{i=1}^d\nu_i^2 \\ & \;\ge\; \frac{1}{2}\Bigl(\sum_{i=1}^d|\nu_i|\Bigr)^2
=\frac{1}{2}\|\nu\|_1^2.
\end{align*}
Hence,
\( 
\nu^\top\nabla^2\psi(y)\,\nu
\;\ge\;
\frac{1}{2}\|\nu\|_1^2,
\) 
and thus $\psi$ is $(1/2)$–strongly convex w.r.t.\ $\|\cdot\|_1$.  Then the result follows from the \[  D_{\psi}(y\|z) = \psi(y) - \psi(z) - \langle \nabla\psi(z),y-z \rangle \geq \frac{1}{2} \|y-z\|^2_1 \]
\end{proof}

\begin{proposition}[Curvature on the action simplex under mass stability]
\label{prop:bregman_action}
Let \(y,z\in(0,1]\Delta^d\) with masses
\(\rho:=\Lambda(z)/\Lambda(y)\in[1-\varepsilon,1+\varepsilon]\),
where \(\varepsilon\in(0,\tfrac25)\).
Then, for every \(t\ge1\),
\[
D_{\psi}(z\,\|\,y)
\;\;\ge\;
\frac{1-\varepsilon}{4}\,\|\theta-x\|_1^{\,2}.
\]
\end{proposition}

\begin{proof}
By Proposition~\ref{prop:timevary_bregman},
\[
D_{\psi}(z\|y)
=
(\tilde\alpha -1)\,D_{\log}(\Lambda(z)\|\Lambda(y))
\;+\;
\rho\,\mathrm{KL}(\theta\parallel x)
\;+\;
(1-\rho)\bigl[H(\theta)-H(x)\bigr].
\]
Then, we write: 
\begin{align}
D_\psi(z \,\|\, y)
&\;\ge\;
\beta \log^2 d \left( \log\left( \frac{1}{\rho} \right) + \rho - 1 \right)
+ (1 - \rho) \big( H(\theta) - H(x) \big)
+ \rho\, \mathrm{KL}(\theta \,\|\, x) \nonumber \\
&\ge
\frac{1}{4} \beta \log^2 d \left( 1 \!- \!\frac{1}{\rho} \right)^2
\!+ \!(\rho \!- \!1) \log d\, \sqrt{2 \mathrm{KL}(\theta \| x)}
\!+ \!\frac{\rho^2}{\beta} \mathrm{KL}(\theta \,\|\, x)
\!+ \!\!\left( \rho \!- \!\frac{\rho^2}{\beta} \right) \mathrm{KL}(\theta \| x) \notag\\
&\;\ge\;
\left( \frac{1}{2} \sqrt{ \beta} \log d \left( 1 - \frac{1}{\rho} \right) 
\;\;+\;\;
\rho \sqrt{\frac{\mathrm{KL}(\theta \,\|\, x)}{\beta} } \right)^2
+ \frac{\rho}{2}\, \mathrm{KL}(\theta \,\|\, x) \nonumber \\
&\;\ge\;
\frac{1}{4} (1 - \epsilon)\, \|\theta - x\|_1^2, \nonumber
\end{align}
where the first step is due to the definition of $D_{log}$. The second step is due to fact \( \log(\frac{1}{\rho})+\rho-1 \ge (1-\frac{1}{\rho})^2\) for \(\rho \in [1-\epsilon,1+\epsilon] \) and Lemma \ref{entropydifference} which implies \( (1-\rho)[H(\theta)-H(x)] \geq (\rho-1)\log d\sqrt{2\mathrm{KL}(\theta\|\, x)}\). The third step is due to the fact that $(\rho-\frac{\rho^2}{\beta})>\frac{\rho}{2}$, and the last step is due to Lemma \ref{pinsker}.
\end{proof}

To analyze MG-DLRC-OMWU equivalently, we start the analysis by taking a closer look at \eqref{eq:dlrc_oftrl_y}. To analyze $\mathrm{Reg}_T$, defined in \eqref{def:regret}, we first study the nonnegative regret defined by
\( 
\tilde{\operatorname{Reg}}(T) := \max_{y^* \in [0,1] \Delta^d} \sum_{t=1}^T \langle u^{(t)}, y^* - y^{(t)} \rangle.
\) 

\begin{proposition}
\label{prop:nonnegative_regret}
For any time horizon $T \in \mathbb{N}$, we have
\( 
\tilde{\operatorname{Reg}}(T) = \max\{0, \mathrm{Reg}(T)\}.
\) 
As a result, $\tilde{\operatorname{Reg}}(T) \geq 0$ and $\tilde{\operatorname{Reg}}(T) \geq \mathrm{reg}_{i,h}^T(s)$.
\end{proposition}

\begin{proof}
By definition of the reward signal $u^{(t)} =w_t \Big( \nu^{(t)} - \langle \nu^{(t)}, x^{(t)} \rangle \mathbf{1}_d\Big)$ and the induced action $x^{(t)} = \frac{y^{(t)}}{\langle y^{(t)}, \mathbf{1} \rangle}$, we have:
\begin{align*}
\tilde{\operatorname{Reg}}(T)
&= \max_{y^* \in [0,1] \cap \Delta^d} \sum_{t=1}^T \langle u^{(t)}, y^* - y^{(t)} \rangle \\
&=  \max_{y^* \in [0,1] \cap \Delta^d} \sum_{t=1}^T w_t\left\langle \nu^{(t)} - \langle \nu^{(t)}, x^{(t)} \rangle \mathbf{1}_d, y^* - y^{(t)} \right\rangle \\
&= \max_{y^* \in [0,1] \cap \Delta^d} \sum_{t=1}^T w_t\left( \langle \nu^{(t)}, y^* \rangle - \langle \nu^{(t)}, x^{(t)} \rangle \langle \mathbf{1}_d, y^* \rangle \right).
\end{align*}
Since $y^* \in \Delta^d$ implies $\langle \mathbf{1}_d, y^* \rangle = 1$, the above simplifies to
\[
\tilde{\operatorname{Reg}}(T) \geq \Big( \max_{y^* \in \Delta^d} \sum_{t=1}^T w_t\big ( \langle \nu^{(t)}, y^* \rangle - \langle \nu^{(t)}, x^{(t)} \rangle\big )\Big)  = \frac{\mathrm{reg}_{i,h}^T(s)}{\alpha^1_T}.
\]
On the other hand, we clearly have $\tilde{\operatorname{Reg}}(T) \geq 0$ by choosing $y^* = 0$ as the comparator.
\end{proof}
This proposition is important as it implies that any RVU bounds on $\tilde{\operatorname{Reg}}(t)$ directly translate into nonnegative RVU bounds on $\mathrm{reg}_{i,h}^t(s)$. Now, define for every \(t\ge1\),
\begin{align}
F_t(y)\;:=\;-\eta_t\bigl\langle U^{(t)}+\kappa^{(t)} u^{(t-1)},\,y\bigr\rangle +\psi(y), \quad 
G_t(z)\;:=\;-\eta_t\bigl\langle U^{(t)},z\bigr\rangle + \psi(z),
\end{align}
where $\kappa^{(t)}=\frac{w_t}{w_{t-1}}$.
The lifted OFTRL iterate and its FTRL proxy are respectively given as
\( 
y^{(t)}
=\arg\min_{y\in(0,1]\Delta^d} F_t(y),\) \( z^{(t)}
=\arg\min_{z\in(0,1]\Delta^d} G_t(z)
\). Then, first we present the following lemma:

\begin{lemma} \label{lem:bregman-optimality}
Given any convex function \( F : \Omega \to \mathbb{R} \) defined on the compact set \( \Omega \), the minimizer
\( 
z^* = \arg\min_{z \in \Omega} F(z)
\) 
satisfies
\( 
F(z^*) \le F(z) - D_F(z \parallel z^*) \quad \forall z \in \Omega,
\) 
where \( D_F \) is the Bregman divergence induced by the function \( F \).
\end{lemma}

\begin{proof}
By definition of the Bregman divergence, and first-order optimality conditions, we have
\[
F(z^*) = F(z) - \langle \nabla F(z^*), z - z^* \rangle - D_F(z \parallel z^*). \le F(z) - D_F(z \parallel z^*),
\]
which proves the claim.
\end{proof}
Following, Lemma \ref{lem:bregman-optimality}, we state state the following lemma:

\begin{lemma}[OFTRL one–step inequality with time–varying step-size]
\label{lem:oftrl_basic_tv}
For any \(y\in(0,1]\Delta^d\) and any horizon \(T\ge1\), the following inequality holds.
\begin{align}
\sum_{t=1}^{T}\!\big\langle y-y^{(t)},u^{(t)}\big\rangle
\;&\le\;
\frac{\psi(y)}{\eta_{T+1}}-\frac{\psi\big(y^{(1)}\big)}{\eta_1}
+\sum_{t=1}^{T}\big\langle z^{(t+1)}-y^{(t)},\,u^{(t)}-\kappa^{(t)}u^{(t-1)}\big\rangle
\notag\\[-0.25em]
&\quad
-\sum_{t=1}^{T}\frac{1}{\eta_t}\Big[D_\psi\big(y^{(t)}\big\|z^{(t)}\big)+D_\psi\big(z^{(t+1)}\big\|y^{(t)}\big)\Big]
+\sum_{t=1}^{T}\left[\frac{1}{\eta_t}-\frac{1}{\eta_{t+1}}\right]\,\psi\big(z^{(t+1)}\big)
\label{eq:oftrl_tv_omega_result}
\end{align}
\end{lemma}

\begin{proof}
By Lemma~\ref{lem:bregman-optimality},
\begin{align}
G_t\big(z^{(t)}\big)
&\le G_t\big(y^{(t)}\big)\;-\;D_\psi\big(y^{(t)}\big\|z^{(t)}\big),
\label{eq:bo_gt}\\
F_t\big(y^{(t)}\big)
&\le F_t\big(z^{(t+1)}\big)\;-\;D_\psi\big(z^{(t+1)}\big\|y^{(t)}\big)
\label{eq:bo_ft}
\end{align}
Moreover,
\begin{equation}
G_t\big(y^{(t)}\big)=F_t\big(y^{(t)}\big)+\eta_t\,\big\langle \kappa^{(t)}u^{(t-1)},\,y^{(t)}\big\rangle.
\label{eq:GvsF_samepoint}
\end{equation}
Furthermore, for any $w\in(0,1]\Delta^d$,
\begin{equation}
\frac{1}{\eta_t}F_t(w)
=\frac{1}{\eta_{t+1}}G_{t+1}(w)
+\big\langle u^{(t)}-\kappa^{(t)}u^{(t-1)},\,w\big\rangle
+\Delta_t\,\psi(w),
\qquad \Delta_t:=\frac{1}{\eta_t}-\frac{1}{\eta_{t+1}}\le 0,
\label{eq:bridge_identity}
\end{equation}
which follows by expanding $F_t,G_{t+1}$ and using $U^{(t+1)}=U^{(t)}+u^{(t)}$. Then, dividing \eqref{eq:bo_gt}--\eqref{eq:bo_ft} by $\eta_t$, substituting \eqref{eq:GvsF_samepoint},
and then applying \eqref{eq:bridge_identity} at $w=z^{(t+1)}$:
\begin{align*}
\frac{1}{\eta_t}G_t\!\big(z^{(t)}\big)
&\le
\frac{1}{\eta_t}F_t\big(y^{(t)}\big)
+\big\langle \kappa^{(t)}u^{(t-1)},y^{(t)}\big\rangle
-\frac{1}{\eta_t}D_\psi\big(y^{(t)}\big\|z^{(t)}\big)\\
&\le
\frac{1}{\eta_t}F_t\big(z^{(t+1)}\big)
-\frac{1}{\eta_t}D_\psi\big(y^{(t)}\big\|z^{(t)}\big)
-\frac{1}{\eta_t}D_\psi\big(z^{(t+1)}\big\|y^{(t)}\big)
+\big\langle \kappa^{(t)}u^{(t-1)},y^{(t)}\big\rangle\\
&=
\frac{1}{\eta_{t+1}}G_{t+1}\!\big(z^{(t+1)}\big)
+\big\langle u^{(t)}-\kappa^{(t)}u^{(t-1)},z^{(t+1)}\big\rangle
+\Delta_t\,\psi\big(z^{(t+1)}\big)\\
&\quad
-\frac{1}{\eta_t}\Big[D_\psi\big(y^{(t)}\big\|z^{(t)}\big)
+D_\psi\big(z^{(t+1)}\big\|y^{(t)}\big)\Big]
+\big\langle \kappa^{(t)}u^{(t-1)},y^{(t)}\big\rangle
\end{align*}
Grouping the linear terms yields
\begin{align}
\frac{1}{\eta_t}G_t\big(z^{(t)}\big)
&\le
\frac{1}{\eta_{t+1}}G_{t+1}\big(z^{(t+1)}\big)
+\big\langle y^{(t)},u^{(t)}\big\rangle
+\big\langle z^{(t+1)}-y^{(t)},\,u^{(t)}-\kappa^{(t)}u^{(t-1)}\big\rangle
\notag\\
&\quad
+\Delta_t\,\psi\big(z^{(t+1)}\big)
-\frac{1}{\eta_t}\Big[D_\psi\big(y^{(t)}\big\|z^{(t)}\big)
+D_\psi\big(z^{(t+1)}\big\|y^{(t)}\big)\Big]
\label{eq:per_t_merged}
\end{align}
Summing \eqref{eq:per_t_merged} over $t=1,\dots,T$ telescopes the $G$-terms:
\begin{align*}
\frac{1}{\eta_1}G_1\!\big(z^{(1)}\big)
\le 
\frac{1}{\eta_{T+1}}G_{T+1}\!\big(z^{(T+1)}\big)
+\sum_{t=1}^{T}\!\big\langle y^{(t)},u^{(t)}\big\rangle 
+\sum_{t=1}^{T}\!\big\langle z^{(t+1)}-y^{(t)},\,u^{(t)}-\kappa^{(t)}u^{(t-1)}\big\rangle\notag
& \\+\sum_{t=1}^{T}\Delta_t\,\psi\big(z^{(t+1)}\big)
-\sum_{t=1}^{T}\frac{1}{\eta_t}\Big[D_\psi(\cdot)+D_\psi(\cdot)\Big]
\end{align*}
Since $U^{(1)}=0$, $G_1=\psi$ and $z^{(1)}=y^{(1)}\in\arg\min_{\Omega}\psi$, and thus $\frac{1}{\eta_1}G_1(z^{(1)})=\frac{\psi(y^{(1)})}{\eta_1}$.
By optimality of $z^{(T+1)}$,
\[
\frac{1}{\eta_{T+1}}G_{T+1}\!\big(z^{(T+1)}\big)
\le
\frac{1}{\eta_{T+1}}G_{T+1}(y)
=-\big\langle U^{(T+1)},y\big\rangle+\frac{\psi(y)}{\eta_{T+1}}.
\]
Insert these, and use $\sum_{t=1}^{T}\langle y,u^{(t)}\rangle=\langle U^{(T+1)},y\rangle$ to obtain
\begin{align*}
\sum_{t=1}^{T}\!\big\langle y-y^{(t)},u^{(t)}\big\rangle
\le
\frac{\psi(y)}{\eta_{T+1}}-\frac{\psi\big(y^{(1)}\big)}{\eta_1}+\sum_{t=1}^{T}\Delta_t\,\psi\big(z^{(t+1)}\big)
+\sum_{t=1}^{T}\!\big\langle z^{(t+1)}-y^{(t)},\,u^{(t)}-\kappa^{(t)}u^{(t-1)}\big\rangle& \\
-\sum_{t=1}^{T}\frac{1}{\eta_t}\Big[D_\psi\big(y^{(t)}\big\|z^{(t)}\big)+D_\psi\big(z^{(t+1)}\big\|y^{(t)}\big)\Big]
\end{align*}
which is precisely \eqref{eq:oftrl_tv_omega_result}.
\end{proof}

\newpage

\begin{lemma} \label{lem:Breg1}
We have the following inequality for iterative Bregman divergences of regularizer $\psi$: 
\[
\sum_{t=1}^T \frac{1}{\eta_t}\left[
  D_\psi\big(y^{(t)} \parallel z^{(t)}\big) + D_\psi\big(z^{(t+1)} \parallel y^{(t)}\big)
\right]
\;\ge\;
\sum_{t=1}^T \frac{1}{2\eta_t}
\left(
  \|y^{(t)} - z^{(t)}\|_1^2 + \|z^{(t+1)} - y^{(t)}\|_1^2
\right).
\]
\end{lemma}

\begin{proof}
By repeated application of Proposition \ref{prop:bregman_lifted}, we have for each \( t \in [T] \):
\[\frac{1}{\eta_t}\left[
D_\psi\big(y^{(t)} \parallel z^{(t)}\big) + D_\psi\big(z^{(t+1)} \parallel y^{(t)}\big)\right]
\;\ge\;
\frac{1}{2\eta_t} \|y^{(t)} - z^{(t)}\|_1^2
+ \frac{1}{2\eta_t} \|z^{(t+1)} - y^{(t)}\|_1^2.
\]
Summing this inequality over all \( t = 1 \) to \( T \) concludes the proof.
\end{proof}

\begin{lemma} \label{lem:Breg2}
If  \( \beta \) is large enough (\( \beta \ge 70 \)), then
\[
\sum_{t=1}^T \frac{1}{\eta_t}\left[
  D_\psi\big(y^{(t)} \parallel z^{(t)}\big) + D_\psi\big(z^{(t+1)} \parallel y^{(t)}\big)
\right]
\;\ge\;
\sum_{t=1}^{T-1} \frac{1}{10\eta_t}
\left(
  \|x^{(t+1)} - \theta^{(t+1)}\|_1^2 + \|\theta^{(t+1)} - x^{(t)}\|_1^2
\right).
\]
\end{lemma}

\begin{proof}
By Theorem \ref{thm:sensitivity}, we know the stability ratio
\( 
\rho := \frac{\Lambda(z)}{\Lambda(y)} \in [1 - \epsilon, 1 + \epsilon] \quad \text{with } \epsilon = \frac{2}{5}.
\) 
Using Proposition \ref{prop:bregman_action}, this implies
\[
 D_\psi(z\parallel y) \ge \frac{1}{4}(1 - \epsilon)\|\theta - x\|_1^2.
\]
Applying this with \( z := z^{(t+1)} \), \( y := y^{(t)} \), \( \theta := \theta^{(t+1)} \), and \( x := x^{(t)} \), we obtain:
\begin{align*}
\frac{1}{\eta_t}D_\psi\big(z^{(t+1)} \parallel y^{(t)}\big)
\ge \frac{1}{4\eta_t}(1 - \epsilon)\|\theta^{(t+1)} - x^{(t)}\|_1^2
= \frac{3}{20\eta_t}\|\theta^{(t+1)} - x^{(t)}\|_1^2
& \\> \frac{1}{10\eta_t} \|\theta^{(t+1)} - x^{(t)}\|_1^2.
\end{align*}
Similarly, we get
\[
\frac{1}{\eta_t}D_\psi\big(y^{(t+1)} \parallel z^{(t+1)}\big)
> \frac{1}{10\eta_t} \|x^{(t+1)} - \theta^{(t+1)}\|_1^2.
\]
Combining both and summing over \( t = 1 \) to \( T-1 \) yields the result.
\end{proof}

\noindent\textbf{Theorem \ref{thm:rvu_tv_text}} \textbf{(RVU bound for MG-DLRC-OMWU with time-varying $\eta_t$)}
{\itshape
Let $\beta\ge70$, and assume that the reward signals obey
$\bigl\|\nu^{(t)}\bigr\|_\infty\le H$ for every $t\in[T]$.
Then, the cumulative regret incurred by the \textnormal{inner}
OFTRL process up to horizon $T$ obeys
\begin{align}
\tilde{\operatorname{Reg}}(T)\!:= 
\sum_{t=1}^{T}\Bigl\langle y-y^{(t)},u^{(t)}\Bigr\rangle
\le  
2\|u^{(t)}\|_\infty\;+\;
\frac{\tilde\alpha\log T+2\log d}{\eta_{T+1}}
 \notag \;+\;\sum_{t=1}^{T}\eta_t\|u^{(t)}-\kappa^{(t)}u^{(t-1)}\|_\infty^2 
&  
\\\;-\;
\frac1{20}\sum_{t=1}^{T-1}
       \frac{\|x^{(t+1)}-x^{(t)}\|_1^{2}}{\eta_t}
\end{align}
}
\begin{proof}
For an arbitrary $y\in\Omega$, define the smoothed comparator
$y' := \tfrac{T-1}{T}\,y+\tfrac{1}{T}\,y^{(1)}\in\Omega$
(where $y^{(1)}:=\arg\min_{y\in\Omega}\psi_1(y)$). Then
\begin{equation}\label{eq:reg_decomp}
\sum_{t=1}^{T}\!\bigl\langle y-y^{(t)},u^{(t)}\bigr\rangle
  \;\leq\;
2\|u^{(t)}\|_\infty +\sum_{t=1}^{T}\bigl\langle y'-y^{(t)},u^{(t)}\bigr\rangle.
\end{equation}
Lemma~\ref{lem:oftrl_basic_tv} with $y=y'$ yields
\begin{align}
\sum_{t=1}^{T}\!\bigl\langle y'-y^{(t)},u^{(t)}\bigr\rangle
&\;\le\;
\underbrace{\frac{\psi(y')}{\eta_{T+1}}-\frac{\psi_1\!\bigl(y^{(1)}\bigr)}{\eta_1}}_{(I)}
\;+\;
\underbrace{\sum_{t=1}^{T}
  \bigl\langle z^{(t+1)}-y^{(t)},\,u^{(t)}-\kappa^{(t)}u^{(t-1)}\bigr\rangle}_{(II)}
\nonumber\\
&\quad-\;
\underbrace{\sum_{t=1}^{T} \frac{1}{\eta_t}
  \Bigl[D_{\psi}\!\bigl(y^{(t)}\!\parallel z^{(t)}\bigr)
       +D_{\psi}\!\bigl(z^{(t+1)}\!\parallel y^{(t)}\bigr)\Bigr]}_{(III)} +\underbrace{\sum_{t=1}^{T}\left[\frac{1}{\eta_t}-\frac{1}{\eta_{t+1}}\right]\,\psi\big(z^{(t+1)}\big)}_{(IV)}.
\label{eq:three_terms}
\end{align}

\noindent

\newcommand{\Lam}{\Lambda}
\newcommand{\one}{\mathbf{1}}

We aim to bound the difference
\( 
(I)
\). Let $\Lam(y):=\sum_{k=1}^d y[k]$ and $H(y):=\sum_{k=1}^d y[k]\log y[k]$.
Recall $\psi_t(y) = -\tilde{\alpha}\,\log\Lam(y) + \frac{1}{\Lam(y)}H(y)$ and set
\( 
S' := \Lam(y'), \) \( S_1 := \Lam\bigl(y^{(1)}\bigr) = 1.
\) Then,
\begin{align}
(I)
&= \frac{\psi_T(y')}{\eta_{T+1}} - \frac{\psi_1(y^{(1)})}{\eta_1}
= \underbrace{-\frac{\tilde{\alpha}}{\eta_{T+1}}\log S'}_{\text{(A)}}
\;+\; \underbrace{\frac{1}{\eta_{T+1}S'}\,H(y')}_{\text{(B)}}
\;+\; \underbrace{\frac{\tilde{\alpha}}{\eta_1}\log S_1}_{\text{(C)}}
\;-\; \underbrace{\frac{1}{\eta_1 S_1}\,H\bigl(y^{(1)}\bigr)}_{\text{(D)}} .
\label{eq:I-4-terms}
\end{align}
By linearity of $\Lam$,
\(
S'=\Lam(y')=\frac{T-1}{T}\,\Lam(y) + \frac{1}{T}\,\Lam(y^{(1)}).
\)
Since $y\in\Omega$ implies $0\le \Lam(y)\le 1$ and $\Lam(y^{(1)})=1$, we obtain
\[
\frac{1}{T} \;\le\; S' \;\le\; 1
\quad\Longrightarrow\quad
0 \;\le\; -\log S' \;\le\; \log T.
\]
Hence
\(
\text{(A)} \;=\; -\frac{\tilde{\alpha}}{\eta_{T+1}}\log S' \;\le\; \frac{\tilde{\alpha}\log T}{\eta_{T+1}}.
\label{eq:A-bound}
\)
Then, write $p'_k:=y'[k]/S'$ so that $\sum_k p'_k=1$ and
\[
\frac{H(y')}{S'} = \log S' + \sum_{k=1}^d p'_k \log p'_k \;\le\; \log S' \;\le\; 0.
\]
Thus
\(
\text{(B)} = \frac{H(y')}{\eta_{T+1} S'} \le 0.
\label{eq:B-bound}
\)
For $y^{(1)}$, set $p^{(1)}_k:=y^{(1)}[k]/S_1 = y^{(1)}[k]$. Since $\sum_k p^{(1)}_k\log p^{(1)}_k \ge -\log d$,
\[
-H\bigl(y^{(1)}\bigr) = -\sum_k y^{(1)}[k]\log y^{(1)}[k] 
\le \log d,
\]
and with $S_1=1$,
\( 
\text{(D)} = -\frac{1}{\eta_1 S_1}\,H\bigl(y^{(1)}\bigr)
\le \frac{\log d}{\eta_1}.
\label{eq:D-bound}
\)
Finally,
\( 
\text{(C)} = \frac{\tilde{\alpha}}{\eta_1}\log S_1 = 0.
\label{eq:C-bound}
\)
Summing over gives\begin{align} \label{I_label}{\quad(I) \;\le\; \frac{\tilde{\alpha}\log T + \log d}{\eta_{T+1}}.\quad}\end{align}
Now, we bound term  $(II)$). For each $t$, apply Hölder and Young with the \emph{local} step-size~$\eta_t$:
\(
\bigl|\langle a,b\rangle\bigr|
\;\le\;
\frac{\|a\|_1^{2}}{4\eta_t}+\eta_t\|b\|_\infty^{2}.
\) 
Set $a=z^{(t+1)}-y^{(t)}$ and $b=u^{(t)}-\kappa^{(t)}u^{(t-1)}$:
\begin{align}
(II)
&\le
\sum_{t=1}^{T}\Bigl[\frac{\|z^{(t+1)}-y^{(t)}\|_1^{2}}{4\eta_t}   +\eta_t\|u^{(t)}-\kappa^{(t)}u^{(t-1)}\|_\infty^2 \Big]. 
\label{eq:termII}
\end{align}
Next, we bound term $(III)$). Combining Lemma~\ref{lem:Breg1} and Lemma~\ref{lem:Breg2}, we have:
\begin{align}
(III)
&\le
-\sum_{t=1}^{T}\frac{\|y^{(t)}-z^{(t)}\|_1^{2}
                    +\|z^{(t+1)}-y^{(t)}\|_1^{2}}{4\eta_t}
-\sum_{t=1}^{T-1}\frac{\|x^{(t+1)}-\theta^{(t+1)}\|_1^{2}
                      +\|\theta^{(t+1)}-x^{(t)}\|_1^{2}}{20\eta_t} \notag \\&\le -\sum_{t=1}^{T}\frac{\|y^{(t)}-z^{(t)}\|_1^{2}
                    +\|z^{(t+1)}-y^{(t)}\|_1^{2}}{4\eta_t}-\sum_{t=1}^{T-1}\frac{\|x^{(t+1)}-x^{(t)}\|_1^{2}
                     }{20\eta_t},
\label{eq:termIII}
\end{align}
where the last step is due to the triangle inequality. 
Finally, we combine $(II)$ and $(III)$):
\begin{align}
(II)+(III)
&\le
\sum_{t=1}^{T}\eta_t\|u^{(t)}-\kappa^{(t)}u^{(t-1)}\|_\infty^2 
\;-\;
\frac1{20}\sum_{t=1}^{T-1}
       \frac{\|x^{(t+1)}-x^{(t)}\|_1^{2}}{\eta_t}
\label{eq:IIplusIII}
\end{align}
N, we bound term $(IV)$. Write $s:=\Lambda(y)=\sum_k y[k]$ and $p_k:=y[k]/s$ whenever $s>0$. Then, 
\[
\psi(y)\;=\;-(\tilde\alpha-1)\,\log s\;+\;\sum_{k=1}^d p_k\log p_k.
\]
Note that $\sum_k p_k\log p_k\in[-\log d,\,0]$, with the minimum $-\log d$ attained at the uniform
$p_k\equiv 1/d$; the first term $-(\tilde\alpha-1)\log s$ is nonnegative for $s\in(0,1]$ and strictly
decreasing in $s$. Hence the global infimum of $\psi$ over $\Omega$ is achieved at $s=1$ and uniform
$p$, i.e., at $y^\star=(1/d,\dots,1/d)$ with 
\( 
\inf_{y\in\Omega}\psi(y)
\;=\;\psi(y^\star)
\;=\;-\log d.
\label{eq:psi-inf}
\) 
Since $\Delta_t\le 0$ and $\psi\!\bigl(z^{(t+1)}\bigr)\ge\inf_{y\in\Omega}\psi(y)$, we have for each $t$:
\( 
\Delta_t\,\psi\bigl(z^{(t+1)}\bigr)\ \le\ \Delta_t\cdot \inf_{y\in\Omega}\psi(y).
\) 
Summing over $T$ and using $\sum_{t=1}^T\Delta_t=\frac{1}{\eta_1}-\frac{1}{\eta_{T+1}}$ yields
\begin{align}\label{IV_label}
\sum_{t=1}^{T}\Delta_t\,\psi\bigl(z^{(t+1)}\bigr)
\ \le\
\Bigl(\tfrac{1}{\eta_1}-\tfrac{1}{\eta_{T+1}}\Bigr)\,(-\log d)
\ =\
\Bigl(\tfrac{1}{\eta_{T+1}}-\tfrac{1}{\eta_1}\Bigr)\,\log d.
\end{align}
Then, inserting \eqref{eq:reg_decomp},\eqref{I_label},\eqref{IV_label},\eqref{eq:IIplusIII}, we get: 
\[
\tilde{\operatorname{Reg}}(T)
\;\le\;
2\|u^{(t)}\|_\infty\;+\;
\frac{\tilde\alpha\log T+2\log d}{\eta_{T+1}}
\;+\;
\sum_{t=1}^{T}\eta_t\|u^{(t)}-\kappa^{(t)}u^{(t-1)}\|_\infty^2 
\;-\;
\frac1{20}\sum_{t=1}^{T-1}
       \frac{\|x^{(t+1)}-x^{(t)}\|_1^{2}}{\eta_t}
\]
\end{proof}

\begin{lemma}\label{lem:gap_supnorm_H}
Assume that $\|\nu^{(t)}\|_\infty \le H$ for all $t\in[T]$. Then,
\begin{align}
\bigl\|u^{(t)} - \kappa^{(t)}u^{(t-1)}\bigr\|_{\infty}^2\;\le\ w_{t}^2\Big( 6\,\|\nu^{(t)} - \nu^{(t-1)}\|_\infty^2
   \;+\;
   4\,H^2\,\|x^{(t)} - x^{(t-1)}\|_1^2 \Big)
\end{align}
\end{lemma}

\begin{proof}
Since  \(  u^{(t)} = w_t \Big( \nu^{(t)} \;-\;\bigl\langle \nu^{(t)},\,x^{(t)}\bigr\rangle \mathbf{1}_d \Big) \), and $\kappa^{(t)} = \frac{w_t}{w_{t-1}}$,  we have
\begin{align*}
\bigl\|u^{(t)} - \kappa^{(t)}u^{(t-1)}\bigr\|_{\infty}^2
&= \Bigl\|w_t\bigl(\nu^{(t)} - \langle\nu^{(t)},x^{(t)}\rangle\mathbf{1}_d\bigr)
      - w_{t}\bigl(\nu^{(t-1)} - \langle\nu^{(t-1)},x^{(t-1)}\rangle\mathbf{1}_d\bigr)\Bigr\|_{\infty}^2 \\
&\le  w_{t}^2\Bigg( \Bigl\|\nu^{(t)} - \nu^{(t-1)}\Bigr\|_\infty
   \;+\;
   \bigl|\langle\nu^{(t)},x^{(t)}\rangle - \langle\nu^{(t-1)},x^{(t-1)}\rangle\bigr|\Bigg)^2 
   \\
&\le w_{t}^2\Big(2\,\|\nu^{(t)} - \nu^{(t-1)}\|_\infty^2
    \;+\;
    2\,\bigl|\langle\nu^{(t)},x^{(t)}\rangle - \langle\nu^{(t-1)},x^{(t-1)}\rangle\bigr|^2\Big)
    \\
&\le w_{t}^2\Big( 2\,\|\nu^{(t)} \!\!-\! \nu^{(t-1)}\|_\infty^2
   \!\!+\!
   4\bigl|\langle\nu^{(t)},x^{(t)} \!\!- \!x^{(t-1)}\rangle\bigr|^2
   \!\!+\!
   4\bigl|\langle\nu^{(t)}\! \!- \!\nu^{(t-1)},x^{(t-1)}\rangle\bigr|^2\Big)
   \\
&\le w_{t}^2\Big( 2\|\nu^{(t)}\! - \!\nu^{(t-1)}\|_\infty^2
   \!+\!
   4\|\nu^{(t)}\|_\infty^2\,\|x^{(t)} \!- \!x^{(t-1)}\|_1^2
   \!+\!
   4\|\nu^{(t)} \!- \!\nu^{(t-1)}\|_\infty^2 \Big)
   \\
&\le w_{t}^2\Big( 6\,\|\nu^{(t)} - \nu^{(t-1)}\|_\infty^2
   \;+\;
   4\,H^2\,\|x^{(t)} - x^{(t-1)}\|_1^2 \Big)
\end{align*}
We used triangle inequality in the first step, Young’s inequality in the second and third steps, and Hölder’s inequality with $\|\nu^{(t)}\|_\infty\le H$ in the fourth step. In the final step, we leverage the \( \|\nu\|_\infty < H\), and group the terms which concludes the proof. 
\end{proof}

\section{Proof of Theorem \ref{thm:dlrc-regret-text}}
In this appendix, we first present the proof of Lemma~\ref{lem:Per-state_regret}, 
where we establish the non-negative regret bound for each \( (s, h) \) pair. 
We then leverage this per-state regret bound to control the second-order path length 
of the policies and recursively bound the CCE-gap in Theorem~\ref{thm:dlrc-regret-text}, 
which allows us to conclude the convergence result for the CCE-gap.

\noindent\textbf{Lemma \ref{lem:Per-state_regret}} \textbf{(Per-state weighted regret bounds: revised RVU)}
{ \itshape Fix an episodic step $h\!\in[H]$, state $s\!\in\mathcal S$, agent $i\!\in\mathcal N$,
and horizon $T\!\ge2$.
Run Algorithm \ref{alg:markov-dlrc-oftrl-lifted} with a \emph{constant} base learning rate~$\eta>0$
and the usual weights $w_j=\alpha_j^t/\alpha_1^t$,
where $\alpha_t=(H+1)/(H+t)$.
Let $|A_{\max}|:=|\mathcal A_i|$ and let be
$\tilde\alpha$ the constant appearing in the Theorem \ref{thm:rvu_tv_text}.
Then, for every $t\in[T]$,}
\begin{align}
\mathrm{reg}_{i,h}^{t}(s)
\le
&
\frac{2H\!\bigl(\tilde\alpha\log t+2\log |A_{\max}|+6H\eta\bigr)}{\eta\,t}
\!+\!
12\eta\,H^{2}(N\!-\!1)\!
      \sum_{j=2}^{t-1}
      \sum_{k\neq i}\!\alpha_t^j
        \bigl\|
          \pi_{h,k}^{j}\!-\!\pi_{h,k}^{j-1}
        \bigr\|_1^{2}
\notag\\
&+
\frac{\!12\eta H^2(3H\!+\!4N^2)\!}{t}
- 
\frac1{24\,\eta H}
      \sum_{j=2}^{t-1}\alpha_{\,t}^{\,j}\,
        \bigl\|
          \pi_{i,h}^{j}-\pi_{i,h}^{j-1}
        \bigr\|_1^{2}.
\label{eq:new_single}
\end{align}
{ \itshape Moreover, summing (\ref{eq:new_single}) over all agents and choosing \(
\displaystyle
\eta =  \frac{1}{24H\sqrt{H}N}
\) yields }
\begin{align}\label{eq:new_sum}
\sum_{i=1}^{N}\!\mathrm{reg}_{i,h}^{t}(s)\!
\le
&\frac{2HN\bigl(\tilde\alpha\log t+2\log |A_{\max}|+6H\eta\bigr)}{\eta\,t}
+\!
\frac{\!12\eta H^2N(3H\!+\!4N^2)\!}{t} \\& \!-
\frac{1}{48\eta H}
      \!\sum_{i=1}^{N}
      \sum_{j=2}^{t-1}\alpha_t^j
        \bigl\|
          \pi_{i,h}^{j}-\pi_{i,h}^{j-1}
        \bigr\|_1^{2} \notag
\end{align}

\begin{proof}
We fix \((s,h)\) and see each episode index \(j\) as a full‑information normal‑form game with payoff vector
\(\nu_j = [Q_{i,h}^j\pi_{h,-i}^j](s,\cdot)\), and \(  u^{(t)} = w_t \Big( \nu^{(t)} \;-\;\bigl\langle \nu^{(t)},\,x^{(t)}\bigr\rangle \mathbf{1}_d \Big) \). By definition of the weighted regret term we have: 
\begin{align*}
\!\mathrm{reg}_{i,h}^{t}(s) \!&:= \!\!\!\max_{\pi^{\dagger,j}_{i,h} \in \mathcal{A}_i} \!\sum_{j=1}^t\! \alpha_t^j \!\! \left\langle \pi^{\dagger,j}_{i,h}\!-\!\pi^j_{i,h} ,  \!\left[Q_{i,h}^{(j)} \pi_{-i,h}^{(j)}\right](s, \cdot) \! \right\rangle \!= \!\alpha_t^1\!\! \max_{\pi^{\dagger,j}_{i,h} \in \mathcal{A}_i} \!\sum_{j=1}^t \!\!\! \ \left\langle \pi^{\dagger,j}_{i,h}\!-\!\pi^j_{i,h} ,  \!w_j\!\!\left[Q_{i,h}^{(j)} \pi_{-i,h}^{(j)}\right](s, \cdot)  \right\rangle  \\& \le \alpha^1_t \Bigg[ 
      2\|u^{(t)}\|_\infty+
\frac{\tilde\alpha\log t+2\log d}{\eta_{t+1}}
+
\sum_{j=1}^{t}\eta_j\|u^{(j)}-\kappa^{(j)}u^{(j-1)}\|_\infty^2 
-
\frac1{20}\sum_{j=1}^{t-1}
       \frac{\|\pi_{i,h}^{j+1}-\pi_{i,h}^{j}\|_1^{2}}{\eta_j} \Bigg] \\& 
       \le \alpha^1_t \Bigg[ 
      2\|u^{(t)}\|_\infty+
\frac{\tilde\alpha\log t+2\log d}{\eta_{t+1}}
+
\sum_{j=1}^{t}\eta_j w_{j}^2\Big( 6\,\|\nu^{(j)} - \nu^{(j-1)}\|_\infty^2
   +
   4\,H^2\,\|\pi_{i,h}^{j}-\pi_{i,h}^{j-1}\|_1^2 \Big) \\ & 
\;-\;
\frac1{20}\sum_{j=1}^{t-1}
       \frac{\|\pi_{i,h}^{j+1}-\pi_{i,h}^{j}\|_1^{2}}{\eta_j} \Bigg]
  \\& 
       \le \alpha^1_t \!\Bigg[ 
      5Hw_t+4H^2\eta w_1\!+\!
\frac{\tilde\alpha\!\log t\!+\!2\log d}{\eta_{t+1}}
\!+\!
\sum_{j=1}^{t-1}\!6\eta w_{j} \Big\|\!\!\left[Q_{i,h}^{(j)} \pi_{-i,h}^{(j)}\right](s, \cdot) \!- \!\!\left[Q_{i,h}^{(j-1)} \pi_{-i,h}^{(j-1)}\right]\!\!(s, \cdot)\Big\|_\infty^2 \\ & 
-
\sum_{j=1}^{t-1}  \frac{w_{j+1}}{24H\eta}
       \|\pi_{i,h}^{j+1}-\pi_{i,h}^{j}\|_1^{2} \Bigg]     
\end{align*}
where the first step is due to $w_j = \frac{\alpha^j_t}{\alpha^1_t}$ and the second step is due to Theorem \ref{thm:rvu_tv_text} and Proposition \ref{prop:nonnegative_regret}. Furthermore, the third step is due to Lemma \ref{lem:gap_supnorm_H}, and the final step is due to step size $\eta = \frac{1}{24H\sqrt{H}N}$. Since \(1/\eta_{t+1} = w_{t+1}/\eta\), $\alpha^t_t = \alpha_t = \frac{H+1}{H+t} \leq \frac{2H}{t}$ and \(\alpha_t^1 w_t = \alpha_t^t\),
\[
\alpha_t^1\!\left( \frac{\tilde\alpha\log t+2\log |A_{\max}|}{\eta_{t+1}} \right)
  =\frac{\tilde\alpha\log t+2\log |A_{\max}|}{\eta}\,\alpha_{t+1}^{t+1}\gamma
  \;\le\;
  \frac{2H(\tilde\alpha\log t+2\log |A_{\max}|)}{\eta\,t},
\]
where $\gamma = \left(\frac{H+1+t}{t}\right)$ with the order of $H$. Likewise, \( \alpha^1_t(5Hw_t+4H^2\eta w_1)\le 6H\,\alpha_t^1w_t\le \frac{12H^2}{t}\); these become the
first term of (\ref{eq:new_single}). For the utility terms, start with the exact decomposition. Then, taking sup‑norms and using 
\(\|A\,x\|_\infty\le\|A\|_\infty\|x\|_1\):
\begin{align*}
\|\bigl(Q_{i,h}^j - Q_{i,h}^{\,j-1}\bigr)\pi_{h,-i}^j
+Q_{i,h}^{\,j-1}\bigl(\pi_{h,-i}^j - \pi_{h,-i}^{\,j-1}\bigr)\|_\infty
&\le
\|(Q_{i,h}^j - Q_{i,h}^{\,j-1})\,\pi_{h,-i}^j\|_\infty
\\
&
+ \|Q_{i,h}^{\,j-1}\,(\pi_{h,-i}^j - \pi_{h,-i}^{\,j-1})\|_\infty \\
&\le
\alpha_j\,H \;+\; H\,\|\pi_{h,-i}^j - \pi_{h,-i}^{\,j-1}\|_1,
\end{align*}
In the first step, we applied the triangle inequality; and in the second step, we used Hölder’s inequality for the \((\|\cdot\|_\infty, \|\cdot\|_1)\) norm pair; and invoked the Bellman update guarantee \(\|Q_{i,h}^j - Q_{i,h}^{\,j-1}\|_\infty \le \alpha_j H\), along with the bounds \(\|Q_{i,h}^{\,j-1}\|_\infty \le H\) and \(\|\pi_{h,-i}^j\|_1 \le 1\). Next squaring both sides and applying \((a+b)^2\le2a^2+2b^2\), yields:
\[
\|\nu^{(j)} - \nu^{(j-1)}\|_\infty^2
\;\le\;
\bigl(\alpha_jH + H\|\pi_{h,-i}^j - \pi_{h,-i}^{\,j-1}\|_1\bigr)^2
\;\le\;
2(\alpha_jH)^2
\;+\;
2H^2\,\|\pi_{h,-i}^j - \pi_{h,-i}^{\,j-1}\|_1^2.
\]
Hence
\begin{align*}
6\eta\,\alpha_t^1
\sum_{j=1}^{t-1}
w_j\|\nu^{(j)} - \nu^{(j-1)}\|_\infty^2 
&\le
12\eta\,\alpha_t^1
\sum_{j=1}^{t-1}
w_j\Bigl((\alpha_jH)^2 + H^2\|\pi_{h,-i}^j-\pi_{h,-i}^{\,j-1}\|_1^2\Bigr)
\\
&=
12\eta H^2
\sum_{j=1}^{t-1}
\alpha_t^j\,\alpha_j^2
\;+\;
12\eta H^2
\sum_{j=1}^{t-1}
\alpha_t^j\,
\|\pi_{h,-i}^j-\pi_{h,-i}^{\,j-1}\|_1^2.
\end{align*}
By Lemma \ref{lem:A1} we have 
\(\displaystyle\sum_{j=1}^t\alpha_t^j\,\alpha_j^2\le\frac{3H}{t}\),
and utilizing the total‐variation bound \cite{hoeffding1958distinguishability}: 
\begin{align*}
\Bigl\|\pi^j_{h,-i}(\cdot|s)-\pi^{j-1}_{h,-i}(\cdot|s)\Bigr\|_1^2
&=
\Bigl(\sum_{a_{-i}\in A_{-i}}
\Bigl|\prod_{k\neq i}\pi^j_{h,k}(a_k|s)
-\prod_{k\neq i}\pi^{\,j-1}_{h,k}(a_k|s)\Bigr|\Bigr)^2\\
&\le
\Bigl(\sum_{k\neq i}
\bigl\|\pi^j_{h,k}(\cdot|s)-\pi^{\,j-1}_{h,k}(\cdot|s)\bigr\|_1\Bigr)^2\\
&\le
(N-1)\sum_{k\neq i}
\bigl\|\pi^j_{h,k}(\cdot|s)-\pi^{\,j-1}_{h,k}(\cdot|s)\bigr\|_1^2\,, 
\end{align*}
we get,
\begin{align*}
6\eta\,\alpha_t^1
\sum_{j=1}^{t-1}
w_j\|\nu^{(j)} - \nu&^{(j-1)}\|_\infty^2
\;\le\; 
\frac{36\,\eta\,H^3}{t}
\;+\;
12\,\eta\,H^2\, (N-1)
\sum_{j=1}^{t-1}
\alpha_t^j
\sum_{k\neq i}
\bigl\|\pi_{h,k}^j-\pi_{h,k}^{\,j-1}\bigr\|_1^2, \\ & 
\le \!\frac{\!12\eta H^2(3H\!+\!4N^2)\!}{t}\!+\!
12\eta H^2 (N\!-\!1)\!\!
\sum_{j=2}^{t-1}
\!\alpha_t^j
\sum_{k\neq i}
\bigl\|\pi_{h,k}^j-\pi_{h,k}^{\,j-1}\bigr\|_1^2,
\end{align*}
where, we have used the following inequality along with the fact \(\alpha^1_t \leq \frac{1}{t}\) by Lemma \ref{lem:A1}, 
\begin{align*}
12\,\eta\,H^2\, (N-1)
\alpha_t^1
\sum_{k\neq i}
\bigl\|\pi_{h,k}^1-\pi_{h,k}^{\,0}\bigr\|_1^2
\;\le\; \frac{48 \eta(N-1)^2H^2}{t}
\end{align*}
leading to the two middle terms in (\ref{eq:new_single}). For the final term, we have: 
\begin{align*}
-\frac{\alpha_t^1}{24H}\sum_{j=1}^{t-1}
       \frac{\lVert \pi_{i,h}^{j+1}-\pi_{i,h}^{j} \rVert_1^{2}}{\eta_{j+1}}
 =-\frac1{24\eta H}\sum_{j=1}^{t-1}\alpha_t^1 w_{j+1}
       \lVert \pi_{i,h}^{j+1}-\pi_{i,h}^{j}\rVert_1^{2}
 &= -\frac1{24\eta H}\sum_{j=1}^{t-1}\alpha_t^{j+1}
       \lVert
        \pi_{i,h}^{j+1}-\pi_{i,h}^{j}
       \rVert_1^{2},\\ & 
  =  -\frac1{24\,\eta H}\sum_{j=2}^{t}\alpha_t^{j}
       \lVert
        \pi_{i,h}^{j}-\pi_{i,h}^{j-1}
       \rVert_1^{2}
\end{align*}
Then, summing up the upper bounds for all four terms we get: 
\begin{align} \label{eq:regret_solo}
\mathrm{reg}_{i,h}^{t}(s)
\le
&
\frac{2H\!\bigl(\tilde\alpha\log t+2\log |A_{\max}|+6H\eta\bigr)}{\eta\,t}
\!+\!
12\eta\,H^{2}(N\!-\!1)\!
      \sum_{j=2}^{t-1}
      \sum_{k\neq i}\!\alpha_t^j
        \bigl\|
          \pi_{h,k}^{j}\!-\!\pi_{h,k}^{j-1}
        \bigr\|_1^{2}
\notag\\
&+
\frac{\!12\eta H^2(3H\!+\!4N^2)\!}{t}
- 
\frac1{24\,\eta H}
      \sum_{j=2}^{t-1}\alpha_{\,t}^{\,j}\,
        \bigl\|
          \pi_{i,h}^{j}-\pi_{i,h}^{j-1}
        \bigr\|_1^{2}.
\end{align}
Finally, summing over all players we get: 
\begin{align} \label{eq:non-neg}
\sum_{i=1}^{N}\mathrm{reg}_{i,h}^{t}(s)
\;\le\;
&\;
\frac{2HN\bigl(\tilde\alpha\log t+2\log |A_{\max}|+6H\eta\bigr)}{\eta\,t}
\;+\;
\frac{\!12\eta H^2N(3H\!+\!4N^2)\!}{t}
\notag\\
&\;
+\;
12\,\eta\,H^{2}(N-1)^{2}
      \sum_{i=1}^{N}
      \sum_{j=2}^{t-1}\alpha_t^j
        \bigl\|
          \pi_{i,h}^{j}(\cdot|s)-\pi_{i,h}^{j-1}(\cdot|s)
        \bigr\|_1^{2}
\notag\\
&\;-\;
\frac1{24\,\eta H}
      \sum_{i=1}^{N}
      \sum_{j=2}^{t-1}\alpha_t^{j}
        \bigl\|
          \pi_{i,h}^{j}(\cdot|s)-\pi_{i,h}^{j-1}(\cdot|s)
        \bigr\|_1^{2}.
\end{align}
Choosing the learning‐rate such as
\(\eta =   1/\bigl(24\,H\sqrt{H}N\bigr)\)
yields the following inequality 
\begin{align*}
\!\!\sum_{i=1}^{N}\!\mathrm{reg}_{i,h}^{t}(s)\!
\le
&
\frac{\!2HN\bigl(\tilde\alpha\log t\!+\!2\log |A_{\max}|\!+\!6H\eta\bigr)\!}{\eta\,t}
\!+\!
\frac{\!12\eta H^2N(3H\!+\!4N^2)\!}{t}\! \!-\!
\frac{\!\!\sum_{i=1}^{N}\!
      \sum_{j=2}^{t-1}\!\alpha_t^j
        \bigl\|
          \pi_{i,h}^{j}-\pi_{i,h}^{j-1}
        \bigr\|_1^{2}\!}{48\eta H}
       \notag
\end{align*}
which leads to (\ref{eq:new_sum}) and concludes the proof. 
\end{proof}

\noindent\textbf{Theorem \ref{thm:dlrc-regret-text}} \textbf{(Regret Bounds for MG-DLRC-OMWU)}
{\itshape
    If the Algorithm \ref{alg:markov-dlrc-oftrl-lifted}, equivalently Algorithm~\ref{alg:markov-dlrc-via-vupdate}, is run on an N-player episodic Markov game for $T$ iterations with parameters $\beta \geq 70$, $\tilde{\alpha} = \beta\log^2|\mathcal{A}_{\max}| + 2\log|\mathcal{A}_{\max}|+2$, and $\eta = 1/24H\sqrt{H}N$, the output policy $\bar\pi$ satisfies: }
    \[
     \text{CCE-Gap}(\bar\pi) \leq \frac{864H^\frac{7}{2}N\!\bigl(\tilde\alpha\log T+2\log |A_{\max}|+2\bigr)}{\,T}
    \]
\begin{proof}
    We know that the right hand side of the inequality given by \ref{eq:non-neg} in Lemma \ref{lem:Per-state_regret} are guaranteed to be non-negative due to Theorem \ref{thm:rvu_tv_text} and Proposition \ref{prop:nonnegative_regret}. Then, we can write the following inequality for \( 12\eta H^2\!N
      \!\sum_{i=1}^{N}\!
      \sum_{j=2}^{t-1}\!\alpha_t^j
        \bigl\|
          \pi_{i,h}^{j}\!-\!\pi_{i,h}^{j-1}
        \bigr\|_1^{2} \) using $\eta = \frac{1}{24H\sqrt{H}N}$ and inequality \ref{eq:non-neg} : 
        
    \begin{align*}\!\!\!\!12\eta H^2\!N
      \!\sum_{i=1}^{N}\!
      \sum_{j=2}^{t-1}\!\alpha_t^j
        \bigl\|
          \pi_{i,h}^{j}\!-\!\pi_{i,h}^{j-1}
        \bigr\|_1^{2}\!
&\le\!
576\eta^2H^3\!N\!\Bigg[ \!
\frac{2H\!N\!\bigl(\tilde\alpha\!\log t\!+\!2\log \!|A_{\max}|\!\!+\!\!6H\!\eta\bigr)\!}{\eta\,t}
\!+
\!\frac{\!12\eta H^2N(3H\!+\!4N^2)\!}{t}\!\Bigg] \\ & \le \frac{2H\bigl(\tilde\alpha\log t+2\log |A_{\max}|+6H\eta\bigr)}{\eta\,t}
\;+\;
\frac{\!12\eta H^2(3H\!+\!4N^2)\!}{t}
\end{align*}
Then plugging last inequality into \ref{eq:regret_solo} and getting rid of negative terms lead us to the following per-state regret per player: 
\begin{align}
\mathrm{reg}_{i,h}^{t}(s)
& 
\le
\frac{4H\!\bigl(\tilde\alpha\log t+2\log |A_{\max}|+6H\eta\bigr)}{\eta\,t}
+
\frac{\!24\eta H^2(3H\!+\!4N^2)\!}{t}\\ & 
\le
\frac{4H\!\bigl(\tilde\alpha\log t+2\log |A_{\max}|+6H\eta\bigr)}{\eta\,t}
+
\frac{\! 3H^2\!+\!4HN\!}{t}\\ & 
\le 
\frac{96H^\frac{5}{2}N\!\bigl(\tilde\alpha\log t+2\log |A_{\max}|+2\bigr)}{\,t}
\end{align}
where the last step is due to $\eta =  \frac{1}{24H\sqrt{H}N}$. Then, from Lemma \ref{lem:gap_recursion_text}, for $1\!\le\!h\!\le\!H,\;1\!\le\!t\!\le\!T$, we have
\begin{equation}
\label{eq:rec-reg}
\delta^{t}_{h}
\;\le\;
\mathrm{reg}_{i,h}^{t}(s)
\;+\;
\sum_{j=1}^{t}\alpha^{(t)}_{j}\,
\delta^{j}_{h+1},
\qquad\text{and}\qquad
\delta^{t}_{H+1}=0.
\end{equation}
Also due to Lemma \ref{lem:A1} we have $\sum_{j=1}^{t}\alpha^{(t)}_j = 1$ and 
\( 
\sum_{j=1}^{t}\alpha^{j}_{t}\frac{1}{j}
\;\le\;
\Bigl(1+\tfrac1H\Bigr)\frac{1}{t}
\). Define
\( 
\gamma \;:=\; 1+\tfrac1H.
\)  Then, we claim that for every $1\!\le\!h\!\le\!H,\;1\!\le\!t\!\le\!T$,
\begin{equation}
\label{eq:claim-reg}
{\;
\delta^{t}_{h}
\;\le\;
\sum_{h'=h}^{H}
\gamma^{\,2(H-h'+0.5)}\,
\frac{96H^\frac{5}{2}N\!\bigl(\tilde\alpha\log t+2\log |A_{\max}|+2\bigr)}{\,t}
\;}
\end{equation}
holds. We proceed by backward induction on $h$.
\noindent\emph{Base case is $h=H+1$.}  
Because $\delta^{t}_{H+1}=0$ and
\(\gamma^{\,2(H-H-1+0.5)}=\gamma^{-1}\),
\eqref{eq:claim-reg} becomes
\[0\le\gamma^{-1}\frac{96H^\frac{5}{2}N\!\bigl(\tilde\alpha\log t+2\log |A_{\max}|+2\bigr)}{\,t},\] which is true. For the induction step, assume \eqref{eq:claim-reg} holds for level $h+1$. Using \eqref{eq:rec-reg} and the induction hypothesis,
for any $t$ we have
\begin{align*}
\delta^{t}_{h}
&\le
\mathrm{reg}_{i,h}^{t}(s)
+\sum_{j=1}^{t}\alpha^{(t)}_{j}
\sum_{h'=h+1}^{H}
\gamma^{\,2(H-h'+0.5)}\,
\frac{96H^\frac{5}{2}N\!\bigl(\tilde\alpha\log t+2\log |A_{\max}|+2\bigr)}{\,j}
\\
& \le
\mathrm{reg}_{i,h}^{t}(s)
+\sum_{h'=h+1}^{H}\gamma^{\,2(H-h'+0.5)}
\Bigl[\,
\sum_{j=1}^{t}\alpha^{(t)}_{j}\,
\frac{96H^\frac{5}{2}N\!\bigl(\tilde\alpha\log T+2\log |A_{\max}|+2\bigr)}{\,j} \Bigr]\\ & 
\leq \frac{96HH^\frac{5}{2}N\!\bigl(\tilde\alpha\log t\!+\!2\log |A_{\max}|\!+\!2\bigr)}{\,t} \!+ \!\!\!\!\!\sum_{h'=h+1}^{H}\!\!\!\!\!\gamma^{\,2(H-h'+0.5)}
\Bigl[
\sum_{j=1}^{t}\!\alpha^{(t)}_{j}
\frac{96H^\frac{5}{2}N\!\bigl(\tilde\alpha\log T\!+\!2\log |A_{\max}|\!+\!2\bigr)}{\,j}\Bigr]  \end{align*} \begin{align*}    
\delta^{t}_{h} &\leq \frac{96HH^\frac{5}{2}N\!\bigl(\tilde\alpha\log T+2\log |A_{\max}|+2\bigr)}{\,t} \Big[1+ \sum_{h'=h+1}^{H}\gamma^{\,2(H-h'+1)}\Big] \\& 
= \frac{96H^\frac{5}{2}N\!\bigl(\tilde\alpha\log T+2\log |A_{\max}|+2\bigr)}{\,t}  \sum_{h'=h}^{H}\gamma^{\,2(H-h')} \\&
\le  \frac{96H^\frac{5}{2}N\!\bigl(\tilde\alpha\log T+2\log |A_{\max}|+2\bigr)}{\,t}  \sum_{h'=h}^{H}\gamma^{\,2(H-h'+0.5)}
\end{align*}

\noindent
which concludes the induction step. This then leads to; 
\begin{align*}
   \delta^{T}_{1} &\leq \frac{96H^\frac{5}{2}N\!\bigl(\tilde\alpha\log T+2\log |A_{\max}|+2\bigr)}{\,T} \sum_{h'=1}^{H}\Bigl(1+\frac{1}{H}\Bigr)^{\,2(H-h'+0.5)} \\& 
   \le \frac{96H^\frac{7}{2}N\!\bigl(\tilde\alpha\log T+2\log |A_{\max}|+2\bigr)}{\,T} \Bigl(1+\frac{1}{H}\Bigr)^{2H}
   \\&  \le \frac{96e^2H^\frac{7}{2}N\!\bigl(\tilde\alpha\log T+2\log |A_{\max}|+2\bigr)}{\,T}  \leq \frac{864H^\frac{7}{2}N\!\bigl(\tilde\alpha\log T+2\log |A_{\max}|+2\bigr)}{\,T}
\end{align*}
Then, by referring to the property that $\text{CCE-Gap}(\bar\pi) \leq \delta^T_1$, we have
\begin{align} \label{eq:ccegap}
     \text{CCE-Gap}(\bar\pi) \leq \frac{864H^\frac{7}{2}N\!\bigl(\tilde\alpha\log T+2\log |\mathcal{A}_{\max}|+2\bigr)}{\,T}
\end{align}

Now, one can easily see that by Lemma~\ref{lem:dlrc_equivalence} (policy-update equivalence), Corollary~\ref{cor:softmax_lambda_opt}, and Lemma~\ref{lem:q-v-update-equivalence} (value-update equivalence), the policy- and value-iteration steps of Algorithms~\ref{alg:markov-dlrc-oftrl-lifted} and \ref{alg:markov-dlrc-via-vupdate} coincide. Moreover, both algorithms use Algorithm~\ref{alg:rollout} for policy execution. Hence the Algorithms \ref{alg:markov-dlrc-oftrl-lifted} and \ref{alg:markov-dlrc-via-vupdate} are equivalent. Then, the CCE-Gap bound stated in \ref{eq:ccegap} for Algorithm~\ref{alg:markov-dlrc-oftrl-lifted} also holds for Algorithm~\ref{alg:markov-dlrc-via-vupdate}. Thus, the proof is complete.

\end{proof}

\end{document}